\documentclass[12pt]{article}
\usepackage{amsmath, amssymb, amsthm, latexsym}

\def\0{{\mathbf 0}}
\def\1{{\mathbf 1}}

\def\F{{\mathbb F}}

\def\Z{{\mathbb Z}}

\def\Zp8{{\Z_{p^\infty}}}

\newtheorem{thm}{Theorem}[section]
\newtheorem{prop}[thm]{Proposition}
\newtheorem{lem}[thm]{Lemma}
\newtheorem{cor}[thm]{Corollary}
\theoremstyle{definition}

\newtheorem{ex}[thm]{Example}

\date{10/18/2012}

\begin{document}

\newcommand{\comment}[1]{} %For comments

\title{Optimal Subcodes of Self-Dual Codes and Their Optimum Distance Profiles}

\author{
Finley Freibert\\
 Department of Mathematics \\
 Ohio Dominican University \\
%1216 Sunbury Rd.\\
 Columbus, OH 43219, USA\\
{Email: \tt ffreibert@yahoo.com}\\
\and
Jon-Lark Kim\thanks{corresponding author} \\
 Department of Mathematics \\
 Sogang University  \\
 Seoul 121-742, South Korea \\
{Email: \tt  jlkim@sogang.ac.kr}
}

\maketitle

\begin{abstract}

Binary optimal codes often contain optimal or near-optimal subcodes.
In this paper we show that this is true for the family of self-dual codes. One approach
is to compute the optimum distance profiles (ODPs) of linear codes,
which was introduced by Luo,~et.~al.~(2010). One of our main results is the development of general algorithms, called the Chain Algorithms, for
finding ODPs of linear codes. Then we determine the ODPs for the Type
II codes of lengths up to $24$ and the extremal Type II codes of
length $32$, give a partial result of the ODP of the extended
quadratic residue code $q_{48}$ of length 48. We also show that there
does not exist a $[48,k,16]$ subcode of $q_{48}$ for $k \ge 17$, and
we find a first example of a {\em doubly-even} self-complementary
$[48, 16, 16]$ code.

\end{abstract}



%{\bf Abbreviated title}:

\newpage

%%%%%%%%%%%%%%%%%%%%%%%%%%%%%%%%%%%%%%%%%%%%%%%%%%%%%%%%%%%%%%%%%%%%%%%%%%%%%%%%%%%%%%%%%%%%%%%%%%%%%%
%%%%%%%%%%%%%%%%%%%%%%%%%%%%%%%%%%%%%%%%%%%%%%%%%%%%%%%%%%%%%%%%%%%%%%%%%%%%%%%%%%%%%%%%%%%%%%%%%%%%%%
%%%%%%%%%%%%%%%%%%%%%%%%%%%%%%%%%%%%%%%%%%%%%%%%%%%%%%%%%%%%%%%%%%%%%%%%%%%%%%%%%%%%%%%%%%%%%%%%%%%%%%
%%%%%%%%%%%%%%%%%%%%%%%%%%%%%%%%%%%%%%%%%%%%%%%%%%%%%%%%%%%%%%%%%%%%%%%%%%%%%%%%%%%%%%%%%%%%%%%%%%%%%%
\section{Introduction}

%Throughout its development, Algebraic Coding Theory has interacted with the areas of
%Algebra, Algorithm Theory, Combinatorics, Number Theory, and Algebraic Geometry~\cite{PleHuf}.

One of the main problems that has arisen in Coding Theory is the search for optimal codes
with the largest size given a minimum distance or optimal codes with the largest minimum
distance given a size~\cite{HufPle, PleHuf, MacSlo}. There has been extensive work in
this direction~\cite{codetables}. Some well-known families of codes, such as the
Reed-Muller codes or the cyclic codes, contain notable subcodes. However,
comparatively little attention has been paid to the subcodes of an
optimal linear code in general. It is a natural concern to determine which
linear codes contain optimal (or near-optimal) subcodes.
Among linear codes, we suggest self-dual or self-orthogonal codes
since their possible non-zero weights
jump by 2 or 4. Thus there is a possibility to get subcodes
with a large minimum distance.

In fact, self-dual codes have been one of the most active topics in
algebraic coding theory since V. Pless~\cite{Ple_72} started to
classify binary self-dual codes in $1972$. These codes have interesting
connections to groups, $t$-designs, lattices, and theta series~\cite{HufPle, MacSlo, RS}.
Furthermore, many extremal self-dual codes often turn out to be
the best among the linear codes with the same parameters. Nevertheless,
little attention has been paid to the subcodes of self-dual codes.

We plan to construct optimal (self-orthogonal) subcodes of a
given linear (self-dual) code.  In order to construct finite-state
codes, Pollara, Cheung and McEliece~\cite{PolCheMcE} constructed the
first $[24, 5, 12]$ subcode of the binary Golay $[24,12, 8]$ code,
improving a previously known $[24, 5, 8]$ subcode. Maks and
Simonis~\cite{MakSim} have shown that there are exactly two
inequivalent $[32, 11, 12]$ codes in the binary Reed-Muller code
$R(2,5)$ which contain $R(1,5)$ and have the weight set $\{0, 12,
16, 20, 32 \}$.

We show that in many cases optimal subcodes can be obtained by
computing optimum distance profiles (ODPs), a concept
introduced by Luo, Han Vinck, and Chen~\cite{LuoVinChe_2010}. The authors~\cite{LuoVinChe_2010}
considered how to construct and then exclude (or include, respectively) the basis
codewords one by one while keeping a distance profile as large as
possible in a dictionary order (or in an inverse dictionary order, respectively).
Thus fault-tolerant capability is improved by selecting subcodes in
communications and storage systems. The practical applications are
found in WCDMA~\cite{HolTos},~\cite{TanWoo} and address retrieval on
optical media~\cite{DijBagTol}.

In~\cite{CheVin_2010} and~\cite{LuoVinChe_2010}, the authors give results on the ODPs of the binary Hamming $[7,4,3]$ code, the binary and ternary Golay codes, Reed-Solomon codes, the first-order and second order Reed-Muller codes.
Recently, Yan,~et.~al.~\cite{YanZhuLuo} considered the optimum distance profiles of some quasi-cyclic codes and proposed two algorithms, called the ``subcodes traversing algorithm'' and ``supercodes traversing algorithm.'' These algorithms enumerate all subcodes of a given code. Hence they are rather inefficient in finding ODPs of linear codes with a relatively large dimension. Their examples have dimension $10$ only. Hence we ask the following two questions.

\medskip

(i) Is there an interesting class of linear codes whose ODPs are not known yet?

(ii) Is there an efficient algorithm to compute ODPs of linear codes?

\medskip

For question (i), we choose a class of self-dual codes since the structure of these subcodes is surprisingly less known. For question (ii), we propose two full algorithms based on cosets, called the Chain Algorithms and two random algorithms to find ODPs of the codes. These algorithms look at a chain of subcodes of a given code and consider the equivalence of the codes with the same dimension. Hence they are more efficient than the subcodes and supercodes traversing algorithm~\cite{YanZhuLuo}.

\medskip

In this paper, we give the ODPs of Type II self-dual codes of
lengths up to $24$ and the five extremal Type II codes of
length $32$, give a partial result of the ODP of the extended
quadratic residue code $q_{48}$ of length 48. Moreover, we show that each of
the five Type II $[32, 16, 8]$ codes contains the two optimal $[32,
11, 12]$ codes, which was previously known only for the Reed-Muller code $R(2,
5)$. We also construct a $[48, 14, 16]$ code and an optimal $[48, 9,
20]$ code from the extended quadratic residue code $q_{48}$ of
length 48. Both codes are not equivalent to the best known codes of
the same parameters in the Magma database~\cite{Mag}.  We also show
that there does not exist a $[48,k,16]$ subcode $C$ of $q_{48}$ for
$k \ge 17$. We find a first example of a {\bf doubly-even}
self-complementary $[48, 16, 16]$ code. Such a code was previously
not known to exist. Only one singly-even self-complementary
$[48,16,16]$ code was found by A. Kohnert~\cite{Koh}. Similarly we
construct $[72, 29, 16]$, $[72, 23, 20]$ codes which are not
equivalent to the best known codes. Further we construct a new
self-orthogonal $[72, 35, 16]$ code with $A_{16} = 129972$. All our computations were done using Magma~\cite{Mag}.

\section{Preliminaries}

We refer to~\cite{HufPle} for basic definitions and results related
to self-dual codes. All codes in this paper are binary. A {\em
linear $[n,k,d]$ code $C$ of length $n$} is a $k$-dimensional subspace
of  $ GF(2)^n$ with the minimum (Hamming) weight $d$. Two codes over $GF(2)$ are said to be {\em equivalent}
 if they differ only by a permutation
of the coordinates.
 The {\em dual} of $C$,
denoted by $C^{\perp}$ is the set of vectors orthogonal to every
codeword of $C$ under the Euclidean inner product. If $C=C^{\perp}$,
$C$ is called {\em self-dual}. A self-dual code is called {\em Type
II ({\mbox{or}} doubly-even)} if every codeword has weight divisible
by $4$, and {\em Type I ({\mbox{or}} singly-even)} if there exists a
codeword whose weight is congruent to $2 \pmod{4}$.

 Let $C$ be a binary self-dual code of length $n$
and minimum distance $d(C)$. Then $d(C)$ satisfies the
following~\cite{RS}.

$$d(C)  \le \left\{
\begin{array}{ll}
4  \left[\frac{n}{24}\right] +4 , & \mbox{ if } n \ne 22 \pmod{24}, \\
4  \left[\frac{n}{24}\right] +6 ,& \mbox{ if } n  = 22 \pmod{24}.
\end{array}\right.$$
A self-dual code meeting one of the above bounds is called {\em
extremal}.
%A code is called {\it optimal} if it has the highest
%possible minimum distance for its length and dimension.

A subcode $C'$ of a linear code $C$ with minimum
distance $d'=d(C') > d(C)$ is {\em maximal} if
there is no subcode $C''$ such that $C' \subsetneq C'' \subsetneq C$
and $d(C'')=d'$.
%there is no subcode $C'' \ne C$ of $C$
%with $d(C'')=d'$ properly containing $C'$.
Given $d'> d(C)$, the maximum of the dimensions of maximal
subcodes $C'$ with $d(C')=d'$ is called the {\em
maximum dimension with respect to $d'$}. Given $n$ and $k$, a linear
$[n,k,d]$ code is {\em minimum distance optimal} if $d$ is the
largest possible. (Grassl's online table~\cite{codetables} is a good
source for reasonable lengths and dimensions for finite fields of
order up to $9$.) Given $n$ and $d$, a linear $[n,k,d]$ code is {\em dimension optimal} if $k$ is the largest possible~\cite[p. 53]{HufPle}.
We raise the following natural question.
Given a binary self-dual code $C$ and any non-zero weight $d' >
d(C)$, how do we find a subcode with maximum dimension with
respect to $d'$? In general, this question seems very difficult since theoretically we should know all subcodes. On
the other hand, there has been another approach related to this
problem, as described below.

Let $C$ be a binary $[n,k]$ code and let $C_0=C$. A sequence of
linear subcodes of $C$, $C_0 \supset C_1 \supset \cdots \supset
C_{k-1}$ is called a {\em subcode chain}, where the dimension of
$C_i$ is $k-i$ for $i=0, \dots k-1$. (If we let $C_k=\{{\bf 0} \}$ and
$V_i=C_{k-i} (i=0, \cdots, k)$, then $\{ {\bf 0} \}=V_0 \subset V_1 \subset
\cdots \subset V_k=C$ is known as a {\em complete flag}~\cite{MilStu}.)

Let $d_i=d(C_i)$ be the minimum distance of $C_i$. Then the
sequence $d_0 \le d_1 \le \cdots \le d_{k-1}$ is called a {\em
distance profile} of $C$ (see
\cite{CheVin_2010},~\cite{LuoVinChe_2010} for details). A generator
matrix such that its first $k-i$ rows (i.e., the remaining rows
after removing its $i$ rows from the bottom) form a generator matrix
of $C_i$ for $0 \le i \le k-1$, is called a {\em generator matrix
with respect to the distance profile}.

For any two integer sequences of length $k$, $a=a_0,\dots, a_{k-1}$
and $b=b_0, \dots, b_{k-1}$, $a$ is called an {\em upper bound on
$b$ in the dictionary order} if $a$ is equal to $b$ or there is an
integer $t$ such that
\[a_i=b_i {\mbox{  for }} 0 \le i \le t-1, {\mbox{ and }} a_t > b_t.
\]

On the other hand,  $a$ is called an {\em upper bound on $b$ in the inverse dictionary order} if $a$ is equal to $b$ or there is an integer $t$ such that
\[a_i=b_i {\mbox{  for }} t+1 \le i \le k-1, {\mbox{ and }} a_t > b_t.
\]

A distance profile of the linear block code is called the {\em
optimum distance profile} (or {\em ODP} for short) {\em in the
dictionary order}, which is denoted by ODP$^{dic}[C](0)$,
ODP$^{dic}[C](1), \dots,$ ODP$^{dic}[C](k-1)$ if it is an upper
bound on any distance profile of $C$ in the dictionary order.
Similarly, a distance profile of the linear block code is called the
{\em optimum distance profile} (or {\em ODP} for short) {\em in the
inverse dictionary order}, which is denoted by ODP$^{inv}[C](0)$,
ODP$^{inv}[C](1), \dots,$ ODP$^{inv}[C](k-1)$ if it is an upper
bound on any distance profile of $C$ in the inverse dictionary
order. To simplify notations, for a given $[n,k]$ code $C$ we may
use {ODP}$^{dic}[C]_i$= {ODP}$^{dic}[C](k-i)$ (resp.
{ODP}$^{inv}[C]_i$= {ODP}$^{inv}[C](k-i)$) so that we may easily
interpret the corresponding subcode parameters: $[n, i,$
{ODP}$^{dic}[C]_i]$ (resp. $[n, i,$ {ODP}$^{inv}[C]_i]$).
 We also use ODP$[C]$ to denote the optimum minimum distance
profile in both orders. Note that for a given $[n,k]$ code $C$ over
$GF(q)$, the number of its subcode chains~\cite{LuoVinChe_2010} is
\[
\prod_{t=2}^k Q[t, t-1] =\prod_{t=2}^k \frac {q^t -1}{q-1}, \]

\noindent where $Q[t,r]$ is the $q$-ary Gaussian binomial coefficient
$\prod_{j=0}^{r-1} \frac{q^{t-j} -1}{q^{r-j} -1}.$ Hence for large
dimensions it will be very difficult to determine ODP of a linear code by
a brute-force search.

\section{Relation between ODP and the maximum dimension}

The ODP of a code and the maximum dimension with
respect to a minimum distance are related concepts.  Note that the first
minimum distance $d'$ to appear in the ODP in dictionary order corresponds to
a maximal subcode with maximum dimension corresponding to $d'$. However, after this
term, maximal subcodes in the subcode chain do not necessarily
imply the maximum dimension.  This is an observation which follows from
the definition of a maximal subcode and the definition of ODP; we formalize the
theory in the following results.  However, note that given a dimension $k' \le k$
there may be multiple minimum distances $d'$ with respect to which $k'$ is the
maximum dimension.  Therefore for the first proposition we define $d_{k'}$ to be the
maximum of such minimum distances.

\begin{prop}\label{maxODP}
Let $C$ be an $[n,k]$ code. Let $k' \le k$ be given. Define $d_{k'}=\max (\{d' : k'$ is the maximum dimension in $C$ with respect to $d'\})$ and define $d_{opt}$ to be the optimal minimum distance attained among all $[n,k']$ codes (many values available at~\cite{codetables}), then

\[d_{opt} \geq d_{k'} \geq  \max ( \{ {\mbox{ODP}}^{{\mbox{dic}}}[C]_{k'}, {\mbox{ODP}}^{\mbox{inv}}[C]_{k'} \} ).\]

\end{prop}
\begin{proof}
The claim $d_{opt} \geq d_{k'}$ is clear since $d_{opt}$ is the maximum minimum distance possible among all $[n,k']$ codes.
By the definition of $d_{k'}$, if $C$ contains an $[n,k',d']$ subcode, then $d_{k'} \geq d'$.
 %To obtain a contradiction suppose there exists an $[n,k',d']$ subcode $C'$ of $C$ with $d' > d_{k'}$ (1).  Note that by definition $d_{k'} \geq d_{k'+i}$ (2) for $0 \leq i \leq k-k'$.  Now by the definition of $d_{k'}$, $C'$ cannot have maximum dimension with respect to $d'$.  Therefore there exists some $[n,k'+i,d']$ subcode of $C$ where $k'+i$ is the maximum dimension for $d'$.  Hence $d_{k'+i} \geq d'$ (3).  Inequalities (1), (2), and (3) imply the contradiction $d' > d_{k'} \geq d_{k'+i} \geq d'$.
Since ODP$^{dic}[C]_{k_i}$ (respectively ODP$^{inv}[C]_{k_i})$ corresponds to a dimension $k_i$ subcode in the subcode chain having minimum distance ODP$^{dic}[C]_{k_i}$ (respectively ODP$^{inv}[C]_{k_i}$), the preceding claim proves the proposition.
\end{proof}

\begin{cor}
Let $C$ be an $[n,k]$ code. Let $k' \le k$ be given. Define $d_{k'}$ and $d_{opt}$ as above.  If ODP$^{dic}[C]_{k'} = d_{opt}$ or ODP$^{inv}[C]_{k'} = d_{opt}$, then
$$d_{opt} = d_{k'}=\max ( \{ {\mbox{ODP}}^{{\mbox{dic}}}[C]_{k'}, {\mbox{ODP}}^{\mbox{inv}}[C]_{k'} \} ).$$

\end{cor}

The necessity of defining $d_{k'}$, in Proposition~\ref{maxODP}, as a maximum is due to the fact that there may be multiple minimum distances yielding the same maximum dimension.  An example where this occurs is the following:

\bigskip

\begin{ex}
Let $C$ be the [6,3,1] code with the following generator matrix:
 \[G = \left[ \begin{array}{lll}
1 1 & 11 &00 \\
11 & 00  &11 \\
10 & 00  &00
\end{array}
\right]
\]
The maximum dimension with respect to $d_1 = 4$ is 2, due to the fact that the first two rows of $G$ generate a [6,2,4] subcode of $C$ with the following generator matrix:
 \[G_1= \left[ \begin{array}{lll}
11 & 11 &00 \\
11 & 00  &11
\end{array}\right]. \]
 Similarly, the maximum dimension with respect to $d_2 = 3$ is 2; this is obtained by adding the third row of $G$ to each row in $G_1$ which yields a $[6,2,3]$ subcode of $C$ with the following generator matrix:
 \[G_2= \left[ \begin{array}{lll}
01 & 11 &00 \\
01 & 00  &11
\end{array}\right] \]
\end{ex}

\bigskip

Notice that in Proposition~\ref{maxODP} we fix the dimension $k'$; a dual statement where we instead fix the minimum distance is the following.

\begin{prop}\label{fixdj}
Let $C$ be an $[n,k]$ code and let $0 \le j \le k-1$. Suppose $d_{j}$ is a minimum distance appearing as ODP$^{dic}[C]_j$ or ODP$^{inv}[C]_j$. Define $k_{j}$ to be the maximum dimension with respect to $d_{j}$, then $k_{j} \geq j$.
\end{prop}
\begin{proof}
The proof follows directly from the definition of maximal dimension with respect to $d_{j}$, since a subcode with this maximal dimension will have dimension $k_{j}$ which is an upper bound on the dimension of any $[n,j,d_{j}]$ subcode.
\end{proof}

The following proposition is a special case of Proposition~\ref{fixdj}; this proposition states that in fact the first minimum distance in the dictionary order ODP corresponds to a maximal subcode with respect to that minimum distance.

\begin{prop}\label{fixdj_special_case}
Let $C$ be an $[n,k,d]$ code.
Suppose that for some $j$, ODP$^{dic}[C]_j$ is the
first term in ODP greater than $d$. Then $j$ is the maximum
dimension with respect to ODP$^{dic}[C]_j$.
\end{prop}
\begin{proof}
If ODP$^{dic}[C]_j$ is the first term in ODP greater than $d$, then ODP$^{dic}[C]_{j+1} = d$ where $0<j< k$.
Suppose to the contrary that $j$ is greater than the maximum dimension with respect to ODP$^{dic}[C]_j$, then there must exist an $[n,j+1]$ subcode with minimum distance ODP$^{dic}[C]_j$.  This implies ODP$^{dic}[C]_{j+1} =$ ODP$^{dic}[C]_j$ by definition of the dictionary order.  Compiling this information we obtain the contradiction: $d=$ ODP$^{dic}[C]_{j+1}=$ ODP$^{dic}[C]_j > d$.
\end{proof}

Propositions~\ref{maxODP}, \ref{fixdj}, and~\ref{fixdj_special_case} give insight into the relation between maximum dimension subcodes and optimum distance profiles.  If a code contains an optimal subcode (minimum distance optimal, dimension optimal, or both) there are many cases where this subcode appears in the subcode chain involved in an optimum distance profile.  However, this is not always the case as in the following example:
\begin{ex}
Let $C$ be the [6,5,1] code with the following generator matrix:
 \[G = \left[ \begin{array}{lll}
1 1 & 11 &00 \\
11 & 00  &11 \\
10 & 10  &10 \\
10 & 10  &00 \\
10 & 00  &00 \\
\end{array}
\right]
\]
By expurgating weight 1 vectors from $C$ we may obtain $[6,4,2]$ subcodes of $C$.  Since there does not exist a $[6,4,3]$ code (see~\cite{codetables}), we may conclude that {ODP}$^{dic}[C]_4 = 2$.  By examining all $[6,4,2]$ subcodes of $C$ it can be determined that none contain a $[6,3,3]$ subcode, and since no $[6,3,4]$ code exists we obtain {ODP}$^{dic}[C]_3 = 2$.  Finally, there is a unique $[6,2,4]$ code (which has a single non-zero weight of 4); as this code is a subcode of at least one $[6,4,2]$ subcode of $C$, and since there does not exist a $[6,2,5]$ code we may conclude {ODP}$^{dic}[C]_2 = 4$ and {ODP}$^{dic}[C]_1 = 4$.  Therefore the optimum distance profile in dictionary order is {ODP}$^{dic}[C]=[1,2,2,4,4]$.

Using similar arguments the ODP in inverse dictionary order is obtained as {ODP}$^{inv}[C]=[1,2,2,3,5]$.  Notice that the first three rows of $G$ generate an optimal [6,3,3] code (both minimum distance optimal and dimension optimal).  Therefore the maximum dimension with respect to minimum distance $d'=3$ is $k'=3$.  However, the subcodes of dimension 3 appearing in both ODP orders have minimum distance 2.  An explanation for this phenomenon is that all supercodes of the [6,3,3] code in $C$ have minimum distance 1.  This is an example where equality is not possible in Proposition~\ref{maxODP} and in Proposition~\ref{fixdj}.
\end{ex}

\bigskip

%{\bf Rewrite this part}

% In this paper, we consider binary Type II self-dual codes
%because they have relatively a few non-zero weights which are also
%divisible by $4$. Then we determine the optimum distance profiles of
%Type II self-dual codes of lengths $8, 16, 24, 32$ and obtain a
%lower bound on the distance profiles of Type II self-dual codes of
%length $48$. We also construct self-orthogonal codes of length $72$
%with best known parameters.

%Here is the question. Given a binary self-dual code $ C$,
%does there exist a subcode with higher minimum distance than that of
%$ C$? We investigate whether these subcodes are good. Since
%binary doubly-even self-dual codes have weights divisible by $4$ and
%we want good minimum distacne it is natural to consider this class
%although any binary self-dual code can be investigated.

\section{Algorithms based on cosets}

Given an $[n,k,d]$ code $C$ which has small length and dimension it may be relatively easy to examine its subcode structure by a brute force generation of all possible subcodes.  However, as length and dimension increase this method becomes very time consuming; this is why we propose four algorithms based on cosets which are relatively efficient in comparison to the brute force search.  The first two algorithms, called the {\em Chain Algorithms} are useful in the classification in the sense that when applying them we obtain a complete list of inequivalent subcodes (respectively supercodes), with prescribed minimum distance, contained in (respectively containing) the given code $C$; in this way, the redundant cases considered in a brute force search are eliminated.  The remaining two {\em Random Algorithms} are random versions of the Chain Algorithms, and especially useful for very large length and dimension, where the exhaustive search is infeasible.  The Random Algorithms can also give results much faster than the Chain Algorithms since all cases are not considered.

\medskip

\noindent
{\bf (Subcodes) Chain Algorithm I:}  An algorithm to produce all maximal subcodes with maximum dimension $k'$ and minimum distance $d' \ge d$.

\begin{enumerate}

\item Input:
Begin with a binary $[n,k,d]$ code $D$ and a positive integer $d' \ge d$ (such that there exists a codeword of weight $d'$ in $D$).

\item Output: Produce the maximum dimension $k'$ among all maximal subcodes with minimum distance $d'$ and a list of inequivalent maximal subcodes of this dimension and minimum distance $d'$.

\medskip

\begin{enumerate}
\item Initialize the set ${\bf B}_1 = \{ D^\perp \}$.  Begin with $i=1$.
\item Build a set ${\bf B}_{i+1}$ of all inequivalent supercodes of dimension $1$ higher of $C$ for all $C \in {\bf B}_i$.  In order to do this we add coset representatives from ${\F_q^n} / C$ to each code $C$ in ${\bf B}_i$.
\item  Check if $d(C^\perp) = d'$ for any code $C \in {\bf B}_{i+1}$.  If ``No'' for any $C \in {\bf B}_{i+1}$, then repeat step (ii) by increasing $i$ to $i+1$.  If ``Yes'', then output the maximum dimension $k'=k-i+1$ and the set of $[n,k-i+1,d']$ subcodes of $D$.
\end{enumerate}

\end{enumerate}

\noindent
 {\bf (Supercodes) Chain Algorithm II:}  An algorithm to find all $[n,k,d]$ supercodes containing an $[n,k',d']$ code with $d' \ge d$ and  $k \ge k'$\\

\begin{enumerate}

\item Input:
Begin with a set ${\bf C}_{k',d'}$ of inequivalent $[n,k',d']$ codes (respectively self-orthogonal codes) with $k \ge k'$ and $d' \ge d$.

\item Output: For each code $C$ in ${\bf C}_{k',d'}$, produce all $[n,k,d]$ codes (respectively self-orthogonal codes) containing $C$.

\medskip
\begin{enumerate}
\item Begin by building a set of all inequivalent supercodes (respectively self-orthogonal supercodes) of dimension $1$ higher of each code $C$ in ${\bf C}_{k',d'}$ with minimum distance greater than or equal to $d$.  In order to do this we add coset representatives from  ${\F_q^n} / C$ (respectively $C^\perp /C$ if $C$ is self-orthogonal) to each code $C$ in ${\bf C}_{k',d'}$ and keep a set of inequivalent supercodes  ${\bf C}_{k'+1}$ generated in this way.
\item  Repeat the first step, by replacing ${\bf C}_{k',d'}$ with ${\bf C}_{k'+1}$ until the set of inequivalent codes which are generated have dimension $k$.
\item  Stop once dimension $k$ is reached.  For each code $C$ in ${\bf C}_{k',d'}$ output all $[n,k,d]$ supercodes of $C$.
\end{enumerate}

\end{enumerate}

\noindent {\bf Analysis and comparison of our algorithms:}

\comment{
Our Chain Algorithms reduce the complexity of calculation by checking in each round the equivalence of all the codes of the same dimension in {\em chains of codes} obtained from a set of given codes. This is one of the two time consuming steps. Another time consuming step is to consider all coset representatives from ${\F_q^n} / C$.
On the other hand, the algorithms given in Yan,~et.~al.~\cite{YanZhuLuo} construct all subcodes of the same dimension not necessarily in chains of codes. Hence their algorithms are computing more than needed (hence less efficient) in calculating ODPs of linear codes.}
Given an $[n,k]$ code $C$, the search for subcodes of dimension $k'$ may be conceptualized as a search tree with root $C$ and each node of branch distance $b$ from $C$ given by a $[n,k-b]$ subcode.  A brute-force search of the subcodes of dimension $k'$ for an $[n,k]$ code searches through all branches of the search tree up to distance $k-k'$; this search has complexity given by the Gaussian binomial coefficient {\small{$\left[ \begin{array}{c} k \\ k' \end{array} \right]_2$}}.  The Chain Algorithms greatly reduce this search by ``pruning'' the search tree in two manners.  First, we keep only inequivalent subcodes (resp. supercodes) at each branch level (in addition this keeps the search efficient memory-wise).  Second, branches can only extend from subcodes that were preserved in the previous step creating a  {\em chain of subcodes}.  In comparison, the algorithms given in Yan,~et.~al.~\cite{YanZhuLuo} construct all subcodes of the same dimension not necessarily in chains of codes; this method corresponds to searching all nodes at a given branch distance (many of which are redundant).

For example, a brute-force search of the subcodes of dimension $k'$ for an $[n,k]$ code has complexity given by the Gaussian binomial coefficient {\small{$\left[ \begin{array}{c} k \\ k' \end{array} \right]_2$}}.
  In Section~\ref{n=32} for some $[32,16,8]$ codes we determine the maximum dimension subcode with respect to $d=12$ to have dimension 11. A brute-force subcode search (such as the subcodes traversing algorithm in~\cite{YanZhuLuo}) would have to enumerate ${\small{\left[ \begin{array}{c} 16 \\ 11 \end{array} \right]_2}} = 120,843,139,740,969,555 $ subcodes; this task is not feasible.

\medskip

\begin{ex}
 As a more concrete example, we determine the ODPs for the four optimal $[28,7,12]$ self-complementary codes classified in~\cite{DodEnKa}.  These codes are doubly-even with non-zero weights 12,16,28.
We begin with a $[28,3,16]$ constant weight code (meaning the only non-zero weight is 16).  There is only one such code due to the fact that all non-zero codewords must intersect in exactly 8 positions; if the first two basis vectors are fixed, then there is only one possibility (up to coordinate permutation) for the third basis vector.  By adding the all-one vector to the constant weight code we obtain a $[28,4,12]$ code with the following generator matrix:
 \[G_{[28,4,16]} = \left[ \begin{array}{lllllll}
    1 1 1 1 &0 0 0 0 &0 0 0 0 &1 1 1 1 &0 0 0 0 &1 1 1 1 &1 1 1 1\\
    0 0 0 0 &1 1 1 1 &0 0 0 0 &1 1 1 1 &1 1 1 1 &0 0 0 0 &1 1 1 1\\
    0 0 0 0 &0 0 0 0 &1 1 1 1 &1 1 1 1 &1 1 1 1 &1 1 1 1 &0 0 0 0\\
    1 1 1 1 &1 1 1 1 &1 1 1 1 &1 1 1 1 &1 1 1 1 &1 1 1 1 &1 1 1 1
\end{array}
\right]
\]
Applying (Supercodes) Chain Algorithm II to this generator matrix (and only keeping doubly-even supercodes) we obtain all four self-complementary $[28,7,12]$ codes with the following generator matrices:
 \[ \left[ \begin{array}{c}
              G_{[28,4,16]}\\
    0 1 0 0 0 1 0 0 0 1 0 1 1 0 1 0 0 1 0 0 1 0 1 1 1 0 0 1\\
    0 0 1 0 0 1 1 1 0 1 1 1 0 1 1 1 0 0 1 1 1 1 0 0 1 1 0 0\\
    0 0 0 1 0 0 0 1 0 0 0 1 1 1 1 0 0 1 0 1 0 1 0 1 0 0 1 1\\
\end{array}
\right],
\left[ \begin{array}{c}
              G_{[28,4,16]}\\
   0 1 0 0 0 1 0 0 0 1 0 0 1 1 0 1 0 0 1 1 1 0 1 0 1 0 1 0\\
   0 0 1 0 0 1 1 1 0 1 1 1 0 1 1 1 0 0 1 1 1 1 0 0 1 1 0 0\\
   0 0 0 1 0 0 0 1 0 0 0 1 1 1 1 0 0 1 0 1 0 1 0 1 0 0 1 1\\

\end{array}
\right],
\]

 \[ \left[ \begin{array}{c}
              G_{[28,4,16]}\\
    0 1 0 0 0 1 0 0 0 1 0 0 1 0 1 1 0 1 0 1 0 0 1 1 0 1 1 0\\
    0 0 1 0 0 1 1 1 0 1 1 1 0 1 1 1 0 0 1 1 1 1 0 0 1 1 0 0\\
    0 0 0 1 0 0 0 1 0 0 0 1 1 1 1 0 0 1 0 1 0 1 0 1 0 0 1 1\\

\end{array}
\right],
\left[ \begin{array}{c}
              G_{[28,4,16]}\\
    0 1 0 1 0 0 1 1 0 0 1 1 1 0 1 0 0 0 0 0 1 0 0 1 1 0 1 0\\
    0 0 1 1 0 0 0 0 0 0 1 1 0 0 1 1 0 0 1 1 0 0 1 1 0 0 1 1\\
    0 0 0 0 0 1 1 0 0 1 0 1 1 0 1 0 0 1 0 1 1 0 1 0 1 1 0 0\\
\end{array}
\right].
\]

Let $C$ be any $[28,7,12]$ self-complementary code.
Since the $[28,3,16]$ subcode is optimal, in light of Lemma~\ref{fixdj_special_case}, we determine {ODP}$^{dic}[C]_3 = 16$.
As a [28,3,16] subcode cannot contain the all-one vector,
 we determine the ODP in dictionary order:
\[{\mbox{ODP}^{dic}}[C]=[12,12,12,12,16,16,16].
\]

The ODP in inverse order is clear since any supercode of the repetition code, containing a weight 16 vector, must also contain a weight 12 vector.  Hence
\[{\mbox{ODP}^{inv}}[C]=[12,12,12,12,12,12,28].
\]

\end{ex}

\bigskip

We now introduce the random algorithms which are random versions of the above coset algorithms:

\bigskip

\noindent
 {\bf Random (Subcodes) Algorithm I:} An algorithm to search for maximal subcodes

\begin{enumerate}

\item Input: A linear code $C$ with parameters $[n,k,d]$
and $d' > d$ where $A_{d'}$ is non-zero.

\item Output: A maximal subcode $C'$ of $C$ with
$d'$.

\medskip

\begin{enumerate}
\item Take any codeword $x$ from $C$ such that wt$(x) \ge
d'$. Let $C_1=\left<x \right>$.

\item Choose any coset representative $y$ of $C/{C_1}$. Let $C_1 :=\left<y \right>+C_1$. Repeat this until
$d(C_1) = d'$.

\item Repeat (b) until there is no coset representative such that
$d(C_1) = d'$. Let $C':=C_1$.
\end{enumerate}
% It takes about $10$ seconds
%to get the following result using this algorithm.

\end{enumerate}

The below algorithm is somewhat opposite to Random Algorithm I.

\noindent {\bf Random (Supercode) Algorithm II:} An algorithm to search for codes containing
good codes

\begin{enumerate}

\item Input: A (best known) linear code $C_1$ with
parameters $[n,k,d]$ and {\em $d' < d$}.

\item Output: A code $C'$ containing $C_1$ with
$d'$ and {\em $k'
>
  k$}
 (if such a $C'$ exists).

\medskip

\begin{enumerate}

\item Let $C$:=${C_1}^\perp$.

\item Choose any coset representative $y$ of $C/{C_1}$. Let $C_1 :=\left<y \right>+C_1$. Repeat this until
$d(C_1) = d'$.

\item Repeat (b) until there is no coset representative such that
$d(C_1) = d'$. Let $C':=C_1$.

\end{enumerate}

\end{enumerate}

\begin{ex}

Using their traversing algorithms, the authors~\cite{YanZhuLuo} have
determined ODPs of a quasi-cyclic $[48,10,20]$ code $C_{48}$ by
finding all $k$-dimensional subcodes of $C$ which is extensive work.
Using the above Random Algorithms, we have also computed ODPs of
$C_{48}$ in the {\em dictionary} and {\em inverse dictionary} orders
{\em in a minute} as follows:

\[{\mbox{ODP}^{dic}}[C_{48}]=[20, 20, 20, 20, 24, 24, 24, 24, 32, 32],
\]
\[{\mbox{ODP}^{inv}}[C_{48}]=[20, 20, 20, 20, 20, 20, 20, 24, 28, 36
].
\]

\end{ex}

\medskip

\section{ODP of Type II self-dual codes}

In this section, we determine the ODP of binary Type
II codes of lengths up to $24$ and the extremal Type II codes of
length $32$.

\medskip

\subsection{$n=8$}
 For length $n=8$, there is a unique binary Hamming $[8,4,4]$ code $e_8$.
It has two non-zero weights $4$ and $8$. It is clear that there is a
unique subcode $\left< {\bf 1} \right>$ of $e_8$ with $d_4=8$. Hence
\[
{\mbox{ODP}}[e_8]=[4,4,4,8].
\] One generator matrix with respect to the ODP
in the dictionary order is

 \[G(e_{8}) = \left[ \begin{array}{llll}
1 1 & 11 &11 & 11 \\
 \hline
00 &00  &11  &11\\
00 &11  &00  &11 \\
01 &01  &01  &01\\
\end{array}
\right].
\]

%ODPB$[C]_0^{{\mbox{dic}}}=4$, ODPB$[C]_1^{{\mbox{dic}}}, \dots,$ ODPB$[C]_{k-1}^{{\mbox{dic}}}$

%Note that $\left< {\bf 1} \right>$ is a maximal subcode and
%is an optimal code. Its codimension is $4-1=3$.
%Hence we have the following relation: $e_8 / \left< {\bf 1} \right>$.

\subsection{$n=16$}

Next let us consider $n=16$.
There are two Type II $[16,8,4]$ codes, denoted by $d_{16}$ and $2e_8$~\cite{ConPle} (blank represents $0$):

 \[ G(d_{16}) = \left[ \begin{array}{llllllll}
11&11 &  &  &  &  &  &\\
11&   &11&  &  &  &  & \\
11&   &  &11&  &  &  &\\
11&   &  &  &11&  &  & \\
11&   &  &  &  &11&  & \\
11&   &  &  &  &  &11& \\
11&   &  &  &  &  &  &11 \\
\hline
1 & 1 & 1&1 &1 &1 & 1&1 \\
\end{array}
\right]
{\mbox{ and }}
G(2e_8)=G(e_8) \oplus G(e_8)
\]

The next higher weight in $d_{16}$ is $8$. We have constructed a $[16,5,8]$
subcode of $d_{16}$. This subcode is equivalent to the first order
Reed-Muller code $R(1,4)$ and hence is unique up to equivalence~\cite{Til}. As
there is no $[16,6,8]$ code~\cite{codetables}, we know that $k=5$ is the maximum
dimension with respect to $d=8$. Since $R(1,4)$ contains the all-one
vector, we have

\[{\mbox{ODP}}[d_{16}]=[4,4,4,8,8,8,8,16].
\]
 Considering some linear combinations of the rows of $G(d_{16})$, we give below one
generator matrix with respect to the ODP in the dictionary order.

 \[ G'(d_{16}) = \left[ \begin{array}{llllllll}
11&11&11&11& 11 &11 &11 &11 \\
\hline
11&11&11&11&    &   &   &  \\
11&11&  &  & 11&11&   &  \\
%11&11&  &  &   &  & 11&11 \\
11&  &11&  &11 &  & 11&  \\
1 &1 &1 &1 &1 &1 &1 &1 \\
\hline
11&   &  &11&  &  &  &\\
11&   &  &  &  &11&  & \\
11&   &  &  &  &  &  &11 \\
\end{array}
\right]
\]

In a similar manner, we have verified that
$2e_8$ has a maximal $[16,5,8]$ subcode, which is generated
by the first five rows of $G'(d_{16})$. Hence we have

\[{\mbox{ODP}}[2e_8]=[4,4,4,8,8,8,8,16].
\]
 We give below one
generator matrix with respect to the ODP in the dictionary order.

 \[ G(2e_8) = \left[ \begin{array}{llllllll}
 11&11&11&11& 11 &11 &11 &11 \\
\hline
11&11&11&11&    &   &   &  \\
11&11&  &  & 11&11&   &  \\
11&  &11&  &11 &  & 11&  \\
1~~&1~~&1~~&1~~&1~~&1~~&1~~&1~~ \\
\hline
 & & & & 11 & 11 &    &\\
 & & & & 11 &    & 11 &\\
 & & & & 1 & 1 & 1 & 1 \\
\end{array}
\right]
\]

\medskip

%{\bf Question: Classify all $[16,5,8]$ subcodes of $d_{16}$ and $2d_8$.}

\medskip

As a summary, we have

\begin{thm}
\[{\mbox{ODP}}[d_{16}]={\mbox{ODP}}[2e_8]=[4,4,4,8,8,8,8,16].
\]
\end{thm}

\subsection{$n=24$}

 Consider $n=24$. There are exactly nine Type II self-dual codes of length $24$.
 These are denoted by $A24 (2d_{12}), B24 (d_{10}+2e_7), C24 (3d_8), D24 (4d_6), E24 (d_{24}), F24 (6d_4), G24 (g_{24}),
 d_{16}+e_8,$ and $3e_8$ in the notations of~\cite{ConPle},~\cite{PleSlo}.
 The first seven codes are indecomposable and the rest are
 decomposable. Note that $G24 (g_{24})$ represents the binary Golay $[24,12, 8]$
 code.

Pollara,~et.~al.~\cite{PolCheMcE} constructed the first $[24, 5,
12]$ subcode $C_{24}^{5, 12}$ of $g_{24}$, improving a previously
known $[24, 5, 8]$ subcode. Note that $C_{24}^{5, 12}$ is
unique~\cite{Til}, has only two non-zero weights $12$ and $16$, and
has a $[24, 2, 16]$ subcode $C_{24}^{2, 16}$. As $C_{24}^{2, 16}$
satisfies the Griesmer bound, it has a generator matrix of which
each row has weight $16$~\cite{Til},~\cite{HufPle}. Hence it is easy
to see that $C_{24}^{2, 16}$ is unique.

Using this information, Luo, et. al.~\cite{LuoVinChe_2010} have
determined

\[
 \begin{array}{lll}
 {\mbox{ODP}}^{dic}[g_{24}]&=&[8,8,8,8,8,8,8,12,12,12,16,16] \\
  {\mbox{ODP}}^{inv}[g_{24}]&=&[8,8,8,8,8,8,8,8,12,12,12,24].
\end{array}
\]

However, less is known of the subcodes of the other Type II
self-dual codes of length $24$. We have checked that the unique
$[24, 5, 12]$ code is contained in any of the nine Type II codes of
length $24$.

Using (Subcodes) Chain Algorithm I we obtain inequivalent maximal $[24,k',8]$
subcodes of each Type II code of length 24 (with minimum distance 4).  Then applying
(Supercodes) Chain Algorithm II to the unique $[24,5,12]$ code for each Type II
 code of length 24 (with minimum distance 4) we obtain a $[24,k',8]$ code equivalent
to one of the maximal subcodes.
Therefore we determine the ODP in the dictionary
order of the Type II $[24,12,4]$ codes as follows.

%{\bf Which algorithm did you use above?}

%{\bf Find the ODP in the inverse dictionary order. %Finley,
%start from the all one vector. Add vectors to see when
%the min. distance jumps to 12.}

\begin{thm}
\[
 \begin{array}{lll}
 {\mbox{ODP}}^{dic}[2d_{12}]&=&[4,4,4,8,8,8,8,12,12,12,16,16] \\
 {\mbox{ODP}}^{dic}[d_{10}+2e_7]&=&[4,4,4,8,8,8,8,12,12,12,16,16]\\
 {\mbox{ODP}}^{dic}[3d_{8}]&=&[4,4,8,8,8,8,8,12,12,12,16,16]\\
{\mbox{ODP}}^{dic}[4d_{6}]&=&[4,4,8,8,8,8,8,12,12,12,16,16]\\
{\mbox{ODP}}^{dic}[d_{24}]&=&[4,4,4,4,8,8,8,12,12,12,16,16]\\
{\mbox{ODP}}^{dic}[6d_{4}]&=&[4,8,8,8,8,8,8,12,12,12,16,16]\\
 {\mbox{ODP}}^{dic}[d_{16}+ e_8]&=&[4,4,4,8,8,8,8,12,12,12,16,16] \\
 {\mbox{ODP}}^{dic}[3e_8]&=&[4,4,4,8,8,8,8,12,12,12,16,16] \\
\end{array}
\]

\end{thm}

For each Type II $[24,12,4]$ code we apply (Subcodes) Chain Algorithm I to
 the maximal $[24,k',8]$ subcodes (containing the all one vector) to obtain a
$[24,4,12]$ subcode (containing the all one vector).
Therefore we may determine the ODP in the inverse dictionary
order of the Type II $[24,12,4]$ codes as follows.

\begin{thm}
\[
 \begin{array}{lll}
 {\mbox{ODP}}^{inv}[2d_{12}]&=&[4,4,4,8,8,8,8,8,12,12,12,24] \\
 {\mbox{ODP}}^{inv}[d_{10}+2e_7]&=&[4,4,4,8,8,8,8,8,12,12,12,24]\\
 {\mbox{ODP}}^{inv}[3d_{8}]&=&[4,4,8,8,8,8,8,8,12,12,12,24]\\
{\mbox{ODP}}^{inv}[4d_{6}]&=&[4,4,8,8,8,8,8,8,12,12,12,24]\\
{\mbox{ODP}}^{inv}[d_{24}]&=&[4,4,4,4,8,8,8,8,12,12,12,24]\\
{\mbox{ODP}}^{inv}[6d_{4}]&=&[4,8,8,8,8,8,8,8,12,12,12,24]\\
 {\mbox{ODP}}^{inv}[d_{16}+ e_8]&=&[4,4,4,8,8,8,8,8,12,12,12,24] \\
 {\mbox{ODP}}^{inv}[3e_8]&=&[4,4,4,8,8,8,8,8,12,12,12,24] \\
\end{array}
\]

\end{thm}

 Table~\ref{tab:sub-24} gives the maximum dimension with respect to minimum distance $d$ for the Type II length 24 codes.

 \begin{cor}
 For each Type II length 24 code, there are maximum dimension subcodes with respect to $d=8, 12, 16, 24$ (except $20$) that are involved in the subcode chain for the ODP in dictionary order or the inverse order.  Furthermore, each Type II length 24 code contains dimension optimal (and minimum distance optimal) subcodes with parameters $[24,5,12], [24,2,16], [24,1,24]$.
  \end{cor}

 %{\bf Check the above corollary. I fixed this way }

\begin{table}[thb]
 \caption{Subcodes of All Type II codes of $n=24$}
 \label{tab:sub-24}
 \begin{center}
%{\small
\begin{tabular}{|l|c|c|}
\noalign{\hrule height1pt}

Codes & max. dim. & max. dim.    \\
      & with $d=8$  & with $d=12$  \\
       \hline
$2d_{12}$    & 9 & 5   \\
$d_{10}+2e_7$ & 9 & 5  \\
$3d_8$        & 10 & 5 \\
$4d_6$        & 10 & 5 \\
$d_{24}$      & 8 & 5  \\
$6d_{4}$      & 11 & 5 \\
$d_{16}+ e_8$ & 9 & 5  \\
$3e_8$        & 9 & 5  \\
$g_{24}$      & 12  & 5 \\
 \noalign{\hrule height1pt}
\end{tabular}
%}
\end{center}
\end{table}

\subsection{$n=32$}\label{n=32}

%Consider $n=32$.
%Note that $d_{32}$ has $d=4$. There is a maximal subcode with
%parameters $[32,8,8]$. Optimal codes have parameters $[32,8,9]$.
As there are $85$ Type II self-dual codes of length $32$, we focus
on extremal Type II self-dual $[32, 16, 8]$ codes. There are exactly
five Type II self-dual $[32, 16, 8]$ codes, denoted by $C81$ (or
$q_{32}$), $C82$ (or $r_{32}$, $R(2,5)$), $C83$ (or $2g_{16}$),
$C84$ (or $8f_4$), $C85$ ($16f_2$)~in the notation
of~\cite{ConPle},~\cite{ConPleSlo}. Using symplectic geometric
approach, Maks and Simonis~\cite{MakSim} show that the second order
Reed-Muller code $r_{32}$ contains exactly two inequivalent $[32,
11, 12]$ codes, each of which further contains the first order Reed-Muller
$[32, 6, 16]$ code $R(1,5)$. Note that any $[32, 6, 16]$ code is
equivalent to $R(1,5)$. Furthermore, Jaffe~\cite{Jaf} proved using
his language {\tt Split} that there exist exactly two $[32, 11, 12]$
codes. These subcodes have optimal dimensions for each minimum
distance. Hence Chen and Han Vinck~\cite{CheVin_2010} have determined the ODP in the
dictionary order for $r_{32}$ as follows:
\[{\mbox{ODP}}[r_{32}]=[8,8,8,8,8,12,12,12,12,12,16,16,16,16,16,32].
\]

On the other hand, little was known of the subcodes of the other four extremal Type II $[32,16,8]$ codes. We show that they also have the same optimum distance profiles as $r_{32}$ does.

\medskip

Using (Supercodes) Chain Algorithm II with $C_{k',d'}=\{R(1,5) \}$, we independently construct two
inequivalent $[32,11,12]$ codes in $r_{32}$ containing $R(1,5)$,
denoted by $RC_1$ and $RC_2$. We note that dim$(RC_1 \cap RC_2) =10$. Using (Supercodes) Chain Algorithm II, we have checked that each of $RC_1$ and $RC_2$ is a subcode of any of the five Type II $[32, 16, 8]$ codes. We denote the five codes based on $RC_1$ ($RC_2$, respectively) by $C81^1, \dots, C85^1$ ($C81^2, \dots, C85^2$, respectively).

Hence we obtain:

\begin{thm} Each code $C$ of the five Type II $[32, 16, 8]$ codes has
\[{\mbox{ODP}}[C]=[8,8,8,8,8,12,12,12,12,12,16,16,16,16,16,32].
\]
\end{thm}

One
generator matrix for each Type II $[32, 16, 8]$ code with respect to the ODP in the dictionary order is
given in the appendix.

{\small
 \[ RC_1  = \left[ \begin{array}{c}
1 1 1 1 1 1 1 1 1 1 1 1 1 1 1 1 1 1 1 1 1 1 1 1 1 1 1 1 1 1 1 1\\
\hline
0 0 0 0 0 0 0 0 0 0 0 0 0 0 0 0 1 1 1 1 1 1 1 1 1 1 1 1 1 1 1 1\\
0 0 0 0 0 0 0 0 1 1 1 1 1 1 1 1 0 0 0 0 0 0 0 0 1 1 1 1 1 1 1 1\\
0 0 0 0 1 1 1 1 0 0 0 0 1 1 1 1 0 0 0 0 1 1 1 1 0 0 0 0 1 1 1 1\\
0 0 1 1 0 0 1 1 0 0 1 1 0 0 1 1 0 0 1 1 0 0 1 1 0 0 1 1 0 0 1 1\\
0 1 0 1 0 1 0 1 0 1 0 1 0 1 0 1 0 1 0 1 0 1 0 1 0 1 0 1 0 1 0 1\\
\hline
1 0 0 0 0 0 0 1 0 0 0 1 0 1 1 1 0 1 0 0 1 1 0 1 0 0 1 0 0 1 0 0\\
0 1 0 0 0 0 0 1 0 0 0 1 0 1 0 0 0 0 1 0 0 1 1 1 1 0 0 0 1 1 0 1\\
0 0 1 0 0 0 0 1 0 1 0 0 0 1 1 1 0 1 1 1 0 1 0 0 0 0 0 1 0 0 1 0\\
0 0 0 0 1 0 0 1 0 0 0 0 1 0 0 1 0 1 0 1 1 1 0 0 1 0 1 0 0 0 1 1\\
0 0 1 0 0 0 0 1 0 0 0 1 0 0 1 0 0 0 0 1 1 1 0 1 1 1 0 1 0 0 0 1\\
\end{array}
\right],~~
 RC_2 = \left[ \begin{array}{c}
1 1 1 1 1 1 1 1 1 1 1 1 1 1 1 1 1 1 1 1 1 1 1 1 1 1 1 1 1 1 1 1\\
\hline
0 0 0 0 0 0 0 0 0 0 0 0 0 0 0 0 1 1 1 1 1 1 1 1 1 1 1 1 1 1 1 1\\
0 0 0 0 0 0 0 0 1 1 1 1 1 1 1 1 0 0 0 0 0 0 0 0 1 1 1 1 1 1 1 1\\
0 0 0 0 1 1 1 1 0 0 0 0 1 1 1 1 0 0 0 0 1 1 1 1 0 0 0 0 1 1 1 1\\
0 0 1 1 0 0 1 1 0 0 1 1 0 0 1 1 0 0 1 1 0 0 1 1 0 0 1 1 0 0 1 1\\
0 1 0 1 0 1 0 1 0 1 0 1 0 1 0 1 0 1 0 1 0 1 0 1 0 1 0 1 0 1 0 1\\
\hline
1 0 0 0 0 0 0 1 0 0 0 1 0 1 1 1 0 1 0 0 1 1 0 1 0 0 1 0 0 1 0 0\\
0 1 0 0 0 0 0 1 0 0 0 1 0 1 0 0 0 0 1 0 0 1 1 1 1 0 0 0 1 1 0 1\\
0 0 1 0 0 0 0 1 0 1 0 0 0 1 1 1 0 1 1 1 0 1 0 0 0 0 0 1 0 0 1 0\\
0 0 0 0 1 0 0 1 0 0 0 0 1 0 0 1 0 1 0 1 1 1 0 0 1 0 1 0 0 0 1 1\\
0 0 1 0 0 0 0 1 0 0 0 1 0 0 1 0 0 1 1 1 1 0 1 1 0 1 0 0 1 0 0 0\\
\end{array}
\right]
\]
}

 \begin{cor}
 For each extremal Type II length 32 code, there are maximum dimension subcodes with respect to $d=12, 16, 32$ that are involved in the subcode chain for the ODP in dictionary order or the inverse order.  Furthermore, each extremal Type II length 32 code contains dimension optimal (and minimum distance optimal) subcodes with parameters $[32,11,12], [32,6,16], [32,1,32]$.
  \end{cor}

\subsection{$n=48$}\label{sec_n-48}

%We will use the notation {ODP}$^{dic}[C]_i$ (resp.
%{ODP}$^{inv}[C]_i$) defined in the preliminaries.

  The extended QR code $q_{48}$ is a unique $[48,24,12]$ self-dual code. Using Random (Subcodes) Algorithm I, we find that
for $d'=16$, there is a maximal $[48,14,16]$
subcode of $q_{48}$. The best known minimum distance optimal $[48,14]$ code has
$d=16$. (Note that $17$ is the upper bound.) One code is given in Magma.
 We have checked that our code is not equivalent to this code. Similarly, for $d'=20$, there is a maximal $[48,9,20]$
subcode of $q_{48}$. This is minimum distance optimal.
One $[48,9,20]$ code is
given in Magma. We have checked that our $[48,9,20]$ code is not
equivalent to this code.
For $d'=24$, there is a maximal $[48,6,24]$
subcode of $q_{48}$, which is in fact a unique code by~\cite{Til}. This is minimum distance optimal. One code is given in
Magma. We have checked that our code is equivalent to
this code.

%{\bf Finley, check if these subcodes are in a subcode-chain.}

With respect to the inverse dictionary order we have examined some self-complementary subcodes of $q_{48}$.  There is a $[48,5,24]$ self-complementary subcode (note that $k=5$ is the maximum dimension of a $[48,k,24]$ self-complementary subcode since the unique $[48,6,24]$ code does not contain the all-one vector).  There is a maximal $[48,9,20]$ self-complementary subcode containing the $[48,5,24]$ code (note that $k=10$ is the maximum dimension of a $[48,k,20]$ self-complementary subcode).

\begin{lem}\label{MacId}{\rm(\cite[the MacWilliams Identities, p. 129]{MacSlo})} Let $C$ be an $[n,k]$ code and denote $A_w$ and $A_w^{\perp}$ to be the number of codewords of weight $w$ in the code $C$ and $C^{\perp}$ respectively. Then
\[
\sum_{i=0}^n A_i P_w(n,i) = 2^k A_w^{\perp}, ~~{\mbox{for}}~~ 0 \le w \le n,
\]
where $P_w(n,i)=\sum_{j=0}^w (-1)^j \left( \begin{array}{c} i \\ j \end{array} \right)
\left( \begin{array}{c} n-i \\ w-j \end{array} \right)$ is a Krawtchouk polynomial.
\end{lem}

Let $C$ be an $[n,k,d]$ code over $\mathbb F_q$. Let $T$ be a set of
$t$ coordinates. Let $C(T)$ be the set of codewords of $C$ which are
{\bf 0} on $T$. We puncture $C(T)$ on $T$ to get a linear code of
length $n-t$ called the {\em code shortened on T} and denoted by
$C_T$~\cite{HufPle}.

\begin{lem}{\rm(\cite[Theorem 1.5.7 (iii)]{HufPle})} \label{thm_shorten} Let $C$ be an $[n,k,d]$ code over
$\mathbb F_q$. Let $T$ be a set of $t$ coordinates. If $t=d$ and $T$
is the set of coordinates where a minimum weight codeword is
non-zero, then $(C^{\perp})_T$ has dimension $n-d-k+1$.
\end{lem}

Both Lemma~\ref{MacId} and Lemma~\ref{thm_shorten} are useful in determining the non-existence of codes with particular parameters and restricted weight distributions.  These lemmas are invoked to prove the non-existence of particular subcodes of the extended quadratic residue code: $q_{48}$.  Lemma~\ref{MacId} is also applied to determine the possible weight distribution of a putative subcode.

 In what follows, we classify all possible weight distributions of a supposed $[48,10,20]$ self-complementary subcode of $q_{48}$.

\begin{lem}
If $C$ is a self-complementary [48,10,20] subcode of $q_{48}$, then the non-zero codewords of $C$ have weights {20,24,28,48}.
\end{lem}
\begin{proof}
Suppose to the contrary that $C$ has non-zero weights {20,28,48}.  Then clearly $A_{20} := 2^9 -1$.  Using the MacWilliams Identities (Lemma~\ref{MacId}) we obtain the equation $2256+16A_{20}=2^{10} A_2^{\perp}$.  Hence $A_2^{\perp} = \frac{163}{16}$, a contradiction.
\end{proof}

\begin{lem}
If $C$ is a self-complementary [48,10,20] subcode of $q_{48}$, then $d^\perp (C) \neq 2$.
\end{lem}
\begin{proof}
Suppose to the contrary that $d^\perp (C) = 2$. Shortening $C$ on a
minimum weight codeword $x_2$ of $C^\perp$ yields a [46,9,20] code
$C_{46}$ with possible non-zero weights {20,24,28} by
Lemma~\ref{thm_shorten} (here we switched the role of $C$ and
$C^{\perp}$).

Define the following matrices:
\[
 \begin{array}{llllll}
 B &=& [A^\perp_{0}(C_{46}) & A^\perp_{1}(C_{46}) & A^\perp_{2}(C_{46}) & A^\perp_{3}(C_{46})]^T  , \\
 A &=& [A_{0}(C_{46}) & A_{20}(C_{46}) & A_{24}(C_{46}) & A_{28}(C_{46})]^T.
\end{array}
\]
Then the MacWilliams Identities yield the matrix equation $2^9 B = PA$, where

 \[P = \left[ \begin{array}{rrrr}
1 & 1 &1 & 1 \\
46 &6  &-2  &-10\\
1035 &-5  &-21  &27 \\
15180 &-100  &44  &60\\
\end{array}
\right].
\]
By Grassl's table~\cite{codetables} there (respectively) does not exist a [45,9,20] linear code and there does not exist a [44,8,20] linear code, therefore respectively we have $A^\perp_{1}(C_{46})=0$ and $A^\perp_{2}(C_{46})=0$.  Combined with the fact that $A^\perp_{0}(C_{46})=1$ the above matrix equation yields a unique solution of:
\begin{equation}
 \begin{array}{llllll}\label{C46wei}
 A &=& [1 & 243 & 147 & 121]^T.
\end{array}
\end{equation}
The possible weight distribution of $C_{46}$ and $C_{46}^\perp$
follows from (\ref{C46wei}).  In particular,  $d(C_{46}^\perp)=3$
which by shortening $C_{46}$ on a minimum weight codeword of
$C_{46}^\perp$ using Lemma~\ref{thm_shorten} implies the existence of a
$[43,7,20]$ code with non-zero weights {20,24,28}. This is a
contradiction to the classification of [43,7,20] due to Bouyuklieva
and Jaffe~\cite{BouJaffe}.
\end{proof}

\begin{lem}
If $C$ is a self-complementary [48,10,20] subcode of $q_{48}$, then there is one possible weight distribution of $C$:
\[
 \begin{array}{lllll}
A_0 = 1& A_{20} = 348& A_{24}=326& A_{28}=348& A_{48} =1.
\end{array}
\]
\end{lem}
\begin{proof}
Define the following matrices:
\[
 \begin{array}{lllll}
 B &=& [A^\perp_{0} & A^\perp_{2} & A^\perp_{4}]^T  , \\
 A &=& [A_{0} & A_{20} & A_{24}]^T.
\end{array}
\]
Then the MacWilliams Identities along with the fact that $C$ is self-complementary yield the matrix equation $2^{10} B = PA$, where

 \[P = \left[ \begin{array}{rrrr}
2 &2 & 1 \\
2256  &16  &-24\\
389160  &-600  &276 \\
\end{array}
\right].
\]
By the previous lemma $A^\perp_{2}=0$, combined with the fact that $A_{0}=A^\perp_{0}=1$ the above matrix equation yields a unique solution of:
\[
 \begin{array}{lllll}
 A &=& [1 & 348 & 326]^T.
\end{array}
\]
\end{proof}

\begin{lem}\label{selfcomp17}
There does not exist a self-complementary [48,k,16] subcode $C$ of $q_{48}$ for $k \ge 17$.
\end{lem}
\begin{proof}
Suppose a [48,17,16] self-complementary subcode $C$ exists.  The possible non-zero weights of $C$ are {16,20,24,28,32,48}.
Define the following matrices:
\[
 \begin{array}{llllll}
B &=& [A^\perp_{0} & A^\perp_{2} & A^\perp_{4} & A^\perp_{6}]^T  , \\
 A &=& [A_{0} & A_{16} & A_{20} & A_{24}]^T.
\end{array}
\]
Then the MacWilliams Identities along with the fact that $C$ is self-complementary yield the matrix equation $2^{17} B = PA$, where

 \[P = \left[ \begin{array}{rrrr}
      2&        2&        2&        1 \\
    2256&      208&       16&      -24 \\
  389160&       40&     -600&      276 \\
24543024&   -14544&     5616&    -2024 \\
\end{array}
\right].
\]
Isolating the matrix $A$ yields the matrix equation $2^{17} P^{-1} B = A$ where
 \[2^{17} P^{-1} = \left[ \begin{array}{rrrr}
   17/14&   65/224&     3/56&    1/224 \\
  9729/2& 17457/32&    211/8&   -15/32 \\
207552/7&  -1012/7&   -752/7&     12/7 \\
   62040&  -1605/2&      162&     -5/2 \\
\end{array}
\right].
\]
The first row of $2^{17} P^{-1}$ implies
\[
\frac{65}{224} A^\perp_{2} + \frac{3}{56} A^\perp_{4} + \frac{1}{224} A^\perp_{6} = -\frac{3}{14},
\]
which is impossible as $A^\perp_{i} \ge 0$ for all $i$.  Hence no such code $C$ can exist.
\end{proof}

The previous lemmas and example from this section yield the following theorem towards the inverse dictionary order ODP for $q_{48}$.
\begin{thm}
\[{\mbox{ODP}}^{inv}[q_{48}]=[12,12,12,12,12,12,12,12,a_1,a_2,a_3,a_4,a_5,a_6,b,20,20,20,20,24,24,24,24,48]
\]
where $a_i \in \{12,16\}$ and $b \in \{12,16,20\}$.
\end{thm}
\begin{proof}
Since $q_{48}$ contains the all-one vector, the repetition code $[48,1,48]$ must be the one dimensional subcode first appearing in the subcode chain.  By~\cite{Til} there is a unique $[48,6,24]$ code with non-zero weights $24,32$; since this code does not contain the all-one vector it cannot be involved in the inverse dictionary order subcode chain.  Hence $k \le 5$ for a $[48,k,24]$ code involved in the subcode chain.  Applying Random (Supercode) Algorithm II to the $[48,1,48]$ subcode of $q_{48}$ we obtained a subcode chain involving a $[48,5,24]$ code contained in a $[48,9,20]$ subcode of $q_{48}$.  Therefore ${\mbox{ODP}}^{inv}[q_{48}]_i=24$ for $2 \le i \le 5$, and ${\mbox{ODP}}^{inv}[q_{48}]_j=20$ for $6 \le i \le 9$.  The maximum dimension for a $[48,k,20]$ code is $k=10$ by Grassl's table~\cite{codetables}, hence ${\mbox{ODP}}^{inv}[q_{48}]_{10}=b$ for $b \in \{12,16,20\}$ and also ${\mbox{ODP}}^{inv}[q_{48}]_j=a_i$ for $11 \le j \le 16$ and $a_i \in \{12,16\}$.  Finally,  ${\mbox{ODP}}^{inv}[q_{48}]_i=12$ for $17 \le i \le 24$ by Lemma~\ref{selfcomp17}.
\end{proof}

\begin{lem} \label{lem:48-k-16}
There does not exist a [48,k,16] subcode $C$ of $q_{48}$ for $k \ge 17$.
\end{lem}
\begin{proof}
Suppose a [48,17,16] subcode of $C$ exists.  Since the self-complementary case is already considered in Lemma~\ref{selfcomp17}, we only need to examine the case where the maximum weight in $C$ is 36 since the non-zero weights in $q_{48}$ are 12, 16, 20, 24, 28, 32, 36, 48.  Hence the possible non-zero weights of $C$ are {16,20,24,28,32,36}.
Define the following matrices:
\[
 \begin{array}{lllllllll}
B &=& [A^\perp_{0} & A^\perp_{1} & A^\perp_{2} & A^\perp_{3} & A^\perp_{4} & A^\perp_{5} & A^\perp_{6}]^T  , \\
 A &=& [A_{0} & A_{16} & A_{20} & A_{24} & A_{28} & A_{32} & A_{36}]^T.
\end{array}
\]
Then the MacWilliams Identities yield the matrix equation $2^{17} B = PA$, where

 \[P = \left[ \begin{array}{rrrrrrr}
       1 &        1 &        1 &        1 &        1 &        1 &        1 \\
      48 &       16 &        8 &        0 &       -8 &      -16 &      -24 \\
    1128 &      104 &        8 &      -24 &        8 &      104 &      264 \\
   17296 &      304 &     -104 &        0 &      104 &     -304 &    -1736 \\
  194580 &       20 &     -300 &      276 &     -300 &       20 &     7380 \\
 1712304 &    -2672 &      456 &        0 &     -456 &     2672 &   -19800 \\
12271512 &    -7272 &     2808 &    -2024 &     2808 &    -7272 &    25080 \\
\end{array}
\right].
\]
Isolating the matrix $A$ yields the matrix equation $2^{17} P^{-1} B = A$ where
 \[2^{17} P^{-1} = \left[ \begin{array}{rrrrrrr}
   34/21 &    17/21 &   65/168 &      1/6 &     1/14 &     1/42 &    1/168 \\
    4788 &     1698 &   2109/4 &      135 &       23 &        1 &     -3/4 \\
   30000 &     4592 &      -61 &     -312 &      -92 &       -8 &        3 \\
   61360 &      680 &     -965 &      140 &      132 &       20 &       -5 \\
212448/7 & -39488/7 &    158/7 &       96 &   -536/7 &   -160/7 &     30/7 \\
    4482 &    -1239 &   3633/8 &    -81/2 &     19/2 &     25/2 &    -15/8 \\
   272/3 &   -272/3 &     65/3 &    -56/3 &        4 &     -8/3 &      1/3 \\

\end{array}
\right].
\]
The first row of $2^{17} P^{-1}$ implies
\[
\frac{17}{21} A^\perp_{1} + \frac{65}{168} A^\perp_{2} + \frac{1}{6} A^\perp_{3} + \frac{1}{14} A^\perp_{4} + \frac{1}{42} A^\perp_{5} + \frac{1}{168} A^\perp_{6} = -\frac{13}{21},
\]
which is impossible as $A^\perp_{i} \ge 0$ for all $i$.  Hence no such code $C$ can exist.

\end{proof}

%{\bf Finley, What is known about ${\mbox{ODP}}^{dic}[q_{48}]$? Fill
%in the below as many as possible.}

\begin{thm}
\[{\mbox{ODP}}^{dic}[q_{48}]=[12,12,12,12,12,12,12,12,a_1,a_2,16,16,16,16,b_1,b_2,b_3,b_3,c_1,c_2,c_3,c_4,d,e]
\]
where $a_i \in \{12,16\}$, $b_k \in \{16,20\}$, $c_l \in \{16,20,24\}$, $d \in \{16,20,24,28,32\}$, and $e \in \{20,24,28,32,36,48\}$.
\end{thm}
\begin{proof}
Note that the non-zero weights in $q_{48}$ are 12,16,20,24,28,32,36,48.  Therefore by Grassl's table~\cite{codetables} we may deduce the following:\\ \\
${\mbox{ODP}}^{dic}[q_{48}]_j=12 $ for $17 \le j \le 24$, by the previous lemma.\\ \\
${\mbox{ODP}}^{dic}[q_{48}]_j=a_i $ for $15 \le j \le 16$ and $a_i \in \{12,16\}$. \\ \\
${\mbox{ODP}}^{dic}[q_{48}]_j=16 $ for $11 \le j \le 14$, because as mentioned at the beginning of this section, there exists a maximal $[48,14,16]$ subcode of $q_{48}$.\\ \\
${\mbox{ODP}}^{dic}[q_{48}]_j=b_i $ for $7 \le j \le 10$ and $b_i \in \{16,20\}$. \\ \\
${\mbox{ODP}}^{dic}[q_{48}]_j=c_i $ for $3 \le j \le 6$ and $c_i \in \{16,20,24\}$. \\ \\
${\mbox{ODP}}^{dic}[q_{48}]_2=d $ for $d \in \{16,20,24,28,32\}$. \\ \\
${\mbox{ODP}}^{dic}[q_{48}]_1=e $ for $e \in \{20,24,28,32,36,48\}$. \\

Note that 16 is not present for values of $e$ because if so then the $[48,14,16]$ code involved in the subcode chain would have to be a constant weight code.  There does not exist a constant weight code (with weight 16) of dimension greater than 5 by the following reasoning. \\ \\
Suppose there exists a $[48,k,16]$ constant weight code.  Define the following matrices:
\[
 \begin{array}{llll}
B &=& [A^\perp_{0} & A^\perp_{1}]^T  , \\
 A &=& [A_{0} & A_{16}]^T.
\end{array}
\]
Then the MacWilliams Identities yield the matrix equation $2^{k} B = PA$, where

 \[P = \left[ \begin{array}{rr}
       1 &        1  \\
      48 &       16 \\
\end{array}
\right].
\]
Since $A^\perp_{0} = 1 = A_{0}$, then the matrix equation yields the following system:

 \[ \begin{array}{ccc}
    2^k &=&  1 + A_{16}  \\
    2^k A^\perp_{1} &=&  48 + 16A_{16}.
\end{array}\]
Solving for $A_{16}$ in the first equation and substituting into the second equation yields:
\[ \begin{array}{c}
     2^k A^\perp_{1} =  48 + 16(2^k -1).
\end{array}\]
Solving for $A^\perp_{1}$ we obtain:

\[ \begin{array}{c}
     A^\perp_{1} = 2^{5-k} + 16.
\end{array}\]
And finally since $A^\perp_{1} $ is an integer, then $k \le 5$.
\end{proof}

From the previous ODPs that have been found for Type II codes, dimension optimal subcodes are involved in subcode chains.  Therefore we have the following:

\medskip

{\bf Conjecture:}
  A $[48,6,24]$ code is involved in a subcode chain for the ODP in dictionary order.

\medskip

Thus we have the following corollary.

\begin{cor}
If a $[48,6,24]$ code is involved in a subcode chain for the ODP in dictionary order, then
\[{\mbox{ODP}}^{dic}[q_{48}]=[12,12,12,12,12,12,12,12,a_1,a_2,16,16,16,16,b_1,b_2,b_3,b_3,24,24,24,24,32,32]
\]
where $a_i \in \{12,16\}$ and $b_j \in \{16,20\}$.
\end{cor}

We were able to find a doubly-even self-complementary $[48,16,16]$ code with generator matrix $G_{[48,16,16]}$.  Such a code was previously not known to exist. Only one singly-even self-complementary $[48,16,16]$ code was found by A. Kohnert~\cite{Koh}.

 The dual code has minimum distance $d=4$.
 The generator matrixfor this doubly-even self-complementary $[48,16,16]$ code is the following:
{\small
\[ G_{[48,16,16]} = \left[ \begin{array}{c}
1 0 0 1 0 0 0 0 0 0 0 0 0 0 1 0 0 0 1 1 0 0 0 1 0 0 1 0 1 1 1 0 0 0 1 1 1 0 0 1
    0 0 0 1 0 1 0 0 \\
0 1 0 1 0 0 0 0 0 0 0 0 0 0 0 0 0 1 1 1 1 0 0 0 0 0 1 0 0 1 0 0 1 1 0 0 1 1 0 0
    1 0 1 1 0 1 0 0 \\
0 0 1 1 0 0 0 0 0 0 0 0 0 0 1 0 0 0 1 0 1 1 1 1 0 0 0 0 1 0 1 0 1 0 0 1 0 1 0 1
    1 1 0 1 1 1 1 0 \\
0 0 0 0 1 0 0 0 0 0 0 0 0 0 1 0 0 0 1 0 1 0 0 0 1 1 1 1 0 0 1 0 1 1 1 0 1 0 1 1
    1 1 0 0 0 1 1 1 \\
0 0 0 0 0 1 0 0 0 0 0 0 0 0 0 0 0 0 0 1 1 1 1 0 0 1 0 0 0 0 1 0 1 0 1 0 1 0 1 0
    1 0 1 1 0 1 0 1 \\
0 0 0 0 0 0 1 0 0 0 0 0 0 0 1 0 0 1 0 0 1 1 1 0 1 0 1 1 0 0 0 0 0 1 0 1 1 0 0 1
    0 0 0 0 1 1 0 1 \\
0 0 0 0 0 0 0 1 0 0 0 0 0 0 0 0 0 1 1 1 1 1 1 1 1 1 0 0 1 0 1 0 0 0 0 0 0 1 1 1
    1 0 0 0 0 0 0 0 \\
0 0 0 0 0 0 0 0 1 0 0 0 0 0 0 0 0 1 1 0 0 0 0 1 1 0 1 1 0 1 0 0 0 1 1 1 1 1 1 1
    1 0 0 0 0 0 0 0 \\
0 0 0 0 0 0 0 0 0 1 0 0 0 0 0 0 0 0 1 1 0 0 1 1 1 0 0 1 1 1 1 0 0 1 0 0 1 1 0 1
    1 1 1 1 1 1 0 0 \\
0 0 0 0 0 0 0 0 0 0 1 0 0 0 0 0 0 1 0 0 1 1 0 1 1 0 0 1 1 1 1 0 0 0 1 1 0 0 1 1
    1 1 1 1 1 0 1 0 \\
0 0 0 0 0 0 0 0 0 0 0 1 0 0 1 0 0 1 1 1 0 0 1 1 1 1 1 1 1 0 0 0 0 0 1 0 1 0 0 0
    0 0 1 0 1 0 0 0 \\
0 0 0 0 0 0 0 0 0 0 0 0 1 0 1 0 0 0 0 1 0 1 0 0 0 1 1 1 1 0 0 0 0 0 1 1 0 0 0 1
    1 1 0 0 1 1 1 0 \\
0 0 0 0 0 0 0 0 0 0 0 0 0 1 1 0 0 1 1 1 1 0 0 1 1 0 0 0 0 0 0 0 0 1 1 0 0 0 0 1
    1 0 0 1 1 1 1 0 \\
0 0 0 0 0 0 0 0 0 0 0 0 0 0 0 1 0 0 0 0 1 1 0 1 0 1 1 0 0 1 1 0 1 1 0 1 0 1 1 0
    1 0 0 1 1 0 0 0 \\
0 0 0 0 0 0 0 0 0 0 0 0 0 0 0 0 1 0 0 1 0 0 1 0 1 1 1 0 0 1 1 0 1 1 0 1 0 0 0 0
    1 1 1 0 0 1 1 0 \\
0 0 0 0 0 0 0 0 0 0 0 0 0 0 0 0 0 0 0 0 0 0 0 0 0 0 0 0 0 0 0 1 1 1 1 1 1 1 1 1
    1 1 1 1 1 1 1 0 \\
\end{array}
\right]
\]
}

{\bf Open Problem 1:} Determine if the code with generator matrix $G_{[48,16,16]}$ is equivalent to a subcode of $q_{48}$.

\subsection{$n=72$}

Note that $q_{72}$ (the extended quadratic residue code of length 72) is a Type II
$[72,36,12]$ code. Due to the complexity, we use Random (Subcodes) Algorithm I.
  For $d'=16$, there is a maximal $[72,29,16]$ subcode of
$q_{72}$ with $A_{16}=2160$. The best known minimum distance optimal $[72,29]$ code
has $d=16$ (and at most $d \le 21$) with $A_{16}=28417$, given in
Magma. Hence our code is not equivalent to this
code.  For $d'=20$, there is a maximal $[72,23,20]$
 subcode with $A_{20}=3046$.
 The best known minimum distance optimal $[72,23]$ code has $d=20$ (and at most $\le 24$)
 with $A_{20}=7120$
 given in Magma. Hence our code is not equivalent to this code.

We start from a best known linear $[72,31,20]$ code, given
in Magma. Let $C_1$ be this code and let $d'=16 < d=20$. Using Random (Supercode) Algorithm II, we have constructed in a few seconds a {\em doubly-even
self-orthogonal $[72,35,16]$ code} $C'$ containing
$C_1$ with $A_{16}=129972$. It is known from Magma that there is a best known minimum distance code with parameters
$[72,35,16]$. This is a doubly-even self-orthogonal
code with $A_{16}={136116}$. Hence our code {\em is not equivalent} to the known code. We do not know how many doubly-even self-orthogonal $[72,35,16]$ codes exist.

\section{Conclusion}

The optimum distance profile for a linear code (and any code in general) is a relatively new concept developed in~\cite{CheVin_2010} and~\cite{LuoVinChe_2010}.  This area is particularly interesting due to its practical applications.  In this paper we relate the optimum distance profile of a code to the concept of maximal subcodes of high minimum distance.  We develop four algorithms which are highly efficient in comparison to a brute force examination of all subcodes.

The classification of self-dual codes continues to be an extremely active area in coding theory.  A particularly interesting class of self-dual codes is those of Type II which have high minimum distance (called extremal or near-extremal).  It is notable that this class of codes contains famous unique codes:  the extended Hamming $[8,4,4]$ code, the extended Golay $[24,12,8]$ code, and the extended quadratic residue $[48,24,12]$ code.  A long standing open problem in coding theory is to prove the existence or non-existence of a Type II $[72,36,16]$ code. The aim of this paper is to shed light on the structure of this interesting class of codes.  We examine the maximal subcodes and ODPs of Type II codes for lengths up to 32.
Of recent significance is the classification of length 40 Type II codes~\cite{Har2011}.  The examination of these codes would be extensive work as there are 16470 Type II $[40,20,8]$ codes (the highest minimum distance in this case is 8 which is not minimum distance optimal by~\cite{codetables}).  Therefore we examined a more interesting case, the unique Type II code $q_{48}$ of length $48$.
 %In the paper, we gave many partial results towards the ODPs of $q_{48}$.  Thus we propose the open problem:

%\bigskip
%{\bf Open Problem 2:} Determine completely the ODP in both orders for $q_{48}$.
\bigskip

%\bigskip
%{\bf Open Problem 2:} Determine the maximum dimension $k_{16}$ of a subcode of $q_{48}$ with minimum distance $d' =16$ and classify all inequivalent $[48,k_{16},16]$ subcodes of $q_{48}$.

%\bigskip
%We conjecture that $k_{16}=16$ since we have found a $[48,16,16]$ doubly-even (in fact self-complementary) code.

%\bigskip
%{\bf Open Problem 3:} Determine the maximum dimension $k_{20}$ of a subcode of $q_{48}$ with minimum distance $d' =20$ and classify all inequivalent $[48,k_{20},20]$ subcodes of $q_{48}$.

%\bigskip
%We conjecture that $k_{20}=10$ since all Type II codes we have examined contain optimal subcodes.

%\bigskip
%{\bf Open Problem 4:} Determine the ODP in both orders for $q_{48}$.

%\bigskip
%Another interesting open problem is of recent significance due to the classification of length 40 Type II codes~\cite{Har2011}.  This would be extensive work as there are 16470 Type II $[40,20,8]$ codes (the highest minimum distance in this case is 8).

%\bigskip
%{\bf Open Problem 5:} Determine the subcode structure and the ODP in both orders for all Type II $[40,20,8]$ codes.

%In the future, there is hope for improvement upon these algorithms with the aim to illuminate the subcode structure of very large codes.

\vspace{1 in}

%\section*{Appendix}
 \centerline{\Large{{\bf Appendix}}}

%\begin{table}
% \caption{Near Extremal FSD Even $[16,8,4]$ codes}
% \label{tab:fsd16_8_4}
% \begin{center}
%{\small
%\begin{tabular}{|l|c|c|c|c|}
%\noalign{\hrule height1pt}

%Weight Distribution & SD & ID & FSD & Total  \\
%       \hline

%$A_{0,16}= 1, A_{4,12}= 4, A_{6,10}= 96, A_{8}= 54$ & 0 & 1 & 0 & 1\\
%$A_{0,16}= 1, A_{4,12}= 8, A_{6,10}= 80, A_{8}= 78$ & 0 & 5 & 0 & 5\\
%$A_{0,16}= 1, A_{4,12}= 10, A_{6,10}= 72, A_{8}= 90$ & 0 & 7 & 10 & 17\\
%$A_{0,16}= 1, A_{4,12}= 12, A_{6,10}= 64, A_{8}= 102$ & 1 & 25 & 32 & 58\\
%$A_{0,16}= 1, A_{4,12}= 14, A_{6,10}= 56, A_{8}= 114$ & 0 & 12 & 24 & 36\\
%$A_{0,16}= 1, A_{4,12}= 16, A_{6,10}= 48, A_{8}= 126$ & 0 & 11 & 10 & 21\\
%$A_{0,16}= 1, A_{4,12}= 20, A_{6,10}= 32, A_{8}= 150$ & 0 & 4 & 0 &4\\
%$A_{0,16}= 1, A_{4,12}= 28, A_{8}= 198$ & 2 & 0 & 0 & 2\\
%\hline
%{\bf Total} & 3 & 65 & 76 & 144\\

% \noalign{\hrule height1pt}
%\end{tabular}
%}
%\end{center}
%\end{table}

{\small

 \[ C81^1 = \left[ \begin{array}{c}
RC_1\\
\hline
1 0 0 0 0 0 0 0 0 0 0 0 0 0 0 1 0 0 0 1 0 1 1 0 0 0 0 0 1 1 1 0\\
0 1 0 0 0 0 0 0 0 0 0 0 0 0 1 0 0 0 0 1 0 1 0 1 0 0 1 1 0 0 0 1\\
0 0 1 0 0 0 0 0 0 0 0 0 0 0 1 0 0 0 1 1 0 0 0 1 0 1 0 0 1 0 0 1\\
0 0 0 1 0 0 0 0 0 0 0 0 0 0 0 1 0 0 0 1 0 0 0 0 1 1 0 0 1 1 0 1\\
0 0 0 0 1 0 0 0 0 0 0 1 0 0 0 0 0 0 1 0 0 1 0 1 0 0 1 1 0 0 1 0\\
\end{array}
\right]
~~~
 C82^1 = \left[ \begin{array}{c}
 RC_1\\
\hline
1 0 0 0 0 0 0 0 0 0 0 0 0 0 0 1 0 1 1 0 1 1 1 0 0 0 0 1 1 1 1 1\\
0 1 0 0 0 0 0 0 0 0 0 0 0 0 0 1 0 0 1 0 0 0 0 0 0 0 1 1 1 0 1 1\\
0 0 1 0 0 0 0 0 0 0 0 0 0 0 0 1 0 1 1 0 1 1 0 1 1 1 0 1 0 1 0 1\\
0 0 0 1 0 0 0 0 0 0 0 0 0 0 0 1 0 0 0 0 1 1 1 0 0 1 0 0 0 1 0 1\\
0 0 0 0 1 0 0 0 0 0 0 0 0 0 0 1 0 0 0 1 1 0 1 0 0 0 0 1 0 0 1 1\\
\end{array}
\right]
\]
}

{\small

\[ C83^1 = \left[ \begin{array}{c}
 RC_1\\
\hline
1 0 0 0 0 0 0 1 0 0 0 1 0 0 0 1 0 0 1 1 0 1 0 1 0 1 1 0 0 1 1 0\\
0 1 0 0 0 0 0 1 0 0 0 1 0 0 1 0 0 0 1 1 0 1 1 0 0 1 0 1 1 0 0 1\\
0 0 1 0 0 0 0 1 0 0 0 1 0 0 1 0 0 0 0 1 0 0 1 0 0 0 1 0 0 0 0 1\\
0 0 0 1 0 0 0 1 0 0 0 1 0 0 0 1 0 0 1 0 0 0 1 0 0 0 1 0 0 0 1 0\\
0 0 0 0 1 0 0 1 0 0 0 0 0 0 0 0 0 0 0 1 0 1 1 1 0 0 1 0 0 0 1 0\\
\end{array}
\right]
~~~
 C84^1 = \left[ \begin{array}{c}
 RC_1\\
\hline
1 0 0 0 0 0 0 0 0 0 0 0 0 0 0 1 0 0 0 1 0 1 1 0 0 0 0 0 1 1 1 0\\
0 1 0 0 0 0 0 0 0 0 0 0 0 0 1 0 0 0 0 1 0 1 0 1 0 0 1 1 0 0 0 1\\
0 0 1 0 0 0 0 0 0 0 0 0 0 0 1 0 0 0 1 1 0 0 0 1 0 1 0 0 1 0 0 1\\
0 0 0 1 0 0 0 0 0 0 0 0 0 0 0 1 0 0 0 0 0 0 0 1 1 0 1 1 0 1 0 1\\
0 0 0 0 1 0 0 0 0 0 0 1 0 0 0 0 0 0 1 1 0 1 0 0 0 1 0 0 1 0 1 0\\
\end{array}
\right]
\]
}

{\small

\[ C85^1 = \left[ \begin{array}{c}
 RC_1\\
\hline
1 0 0 0 0 0 0 0 0 0 0 0 0 0 0 1 0 0 0 1 0 1 1 0 0 0 0 0 1 1 1 0\\
0 1 0 0 0 0 0 0 0 0 0 0 0 0 0 1 0 1 0 1 0 1 1 1 1 1 0 1 1 0 1 0\\
0 0 1 0 0 0 0 0 0 0 0 0 0 0 0 1 0 0 0 1 1 0 1 0 0 0 1 1 0 1 0 0\\
0 0 0 1 0 0 0 0 0 0 0 0 0 0 0 1 0 0 0 0 0 0 0 1 1 0 1 1 0 1 0 1\\
0 0 0 0 1 0 0 0 0 0 0 0 0 0 0 1 0 0 0 0 0 1 0 0 1 0 0 1 1 0 1 1\\
\end{array}
\right]
~~~
 C81^2 = \left[ \begin{array}{c}
 RC_2\\
\hline
1 0 0 0 0 0 0 0 0 0 0 1 0 0 0 0 0 0 1 0 0 1 0 1 0 0 1 0 1 1 0 0\\
0 1 0 0 0 0 0 0 0 0 0 1 0 0 1 1 0 0 1 0 0 1 1 0 1 1 1 0 1 1 0 0\\
0 0 1 0 0 0 0 0 0 0 0 1 0 0 1 1 0 0 0 0 0 0 1 0 1 0 0 1 0 1 0 0\\
0 0 0 1 0 0 0 0 0 0 0 1 0 0 0 0 0 0 1 0 0 0 1 1 1 1 1 0 1 1 1 1\\
0 0 0 0 1 0 0 0 0 0 0 0 0 0 0 1 0 0 0 1 0 1 1 0 1 1 1 0 1 1 1 1\\
\end{array}
\right]
\]
}

{\small
\[ C82^2 = \left[ \begin{array}{c}
 RC_2\\
\hline
1 0 0 0 0 0 0 1 0 0 0 1 0 1 1 1 0 0 0 1 0 1 1 1 0 1 1 1 1 1 1 0\\
0 1 0 0 0 0 0 1 0 0 0 1 0 1 0 0 0 0 0 1 0 1 0 0 0 1 0 0 0 0 0 1\\
0 0 1 0 0 0 0 1 0 0 0 1 0 0 1 0 0 0 0 1 0 0 1 0 0 0 1 0 0 0 0 1\\
0 0 0 1 0 0 0 1 0 0 0 1 0 0 0 1 0 0 0 1 0 0 0 1 0 0 0 1 0 0 0 1\\
0 0 0 0 1 0 0 1 0 0 0 0 0 1 1 0 0 0 0 0 0 1 1 0 0 0 0 0 1 0 0 1\\
\end{array}
\right]
~~~
 C83^2 = \left[ \begin{array}{c}
 RC_2\\
\hline
1 0 0 0 0 0 0 1 0 0 0 1 0 0 0 1 0 0 1 1 0 1 0 1 0 1 1 0 0 1 1 0\\
0 1 0 0 0 0 0 1 0 0 0 1 0 0 1 0 0 0 1 1 0 1 1 0 0 1 0 1 1 0 0 1\\
0 0 1 0 0 0 0 1 0 0 0 1 0 0 1 0 0 0 0 1 0 0 1 0 0 0 1 0 0 0 0 1\\
0 0 0 1 0 0 0 1 0 0 0 1 0 0 0 1 0 0 1 0 0 0 1 0 0 0 1 0 0 0 1 0\\
0 0 0 0 1 0 0 1 0 0 0 0 0 0 0 0 0 0 0 1 0 1 1 1 0 0 1 0 0 0 1 0\\
\end{array}
\right]
\]
}

{\small
\[ C84^2 = \left[ \begin{array}{c}
 RC_2\\
\hline
1 0 0 0 0 0 0 0 0 0 0 1 0 0 0 0 0 0 1 0 0 1 0 1 0 0 1 0 1 1 0 0\\
0 1 0 0 0 0 0 0 0 0 0 1 0 0 1 1 0 0 1 0 0 1 1 0 1 1 1 0 1 1 0 0\\
0 0 1 0 0 0 0 0 0 0 0 1 0 0 1 1 0 0 0 0 0 0 1 0 1 0 0 1 0 1 0 0\\
0 0 0 1 0 0 0 0 0 0 0 1 0 0 0 0 0 0 1 1 0 0 1 0 0 1 1 0 1 0 0 0\\
0 0 0 0 1 0 0 0 0 0 0 0 0 0 0 1 0 0 0 0 0 1 1 1 0 1 1 0 1 0 0 0\\
\end{array}
\right]
~~~
 C85^2 = \left[ \begin{array}{c}
 RC_2\\
\hline
1 0 0 0 0 0 0 0 0 0 0 1 0 0 0 0 0 0 1 0 0 1 0 1 0 0 1 0 1 1 0 0\\
0 1 0 0 0 0 0 0 0 0 0 1 0 0 0 0 0 1 1 0 0 1 1 1 1 1 1 1 0 1 0 0\\
0 0 1 0 0 0 0 0 0 0 0 1 0 0 0 0 0 1 0 0 0 0 1 1 1 0 0 0 1 1 0 0\\
0 0 0 1 0 0 0 0 0 0 0 1 0 0 0 0 0 0 1 0 0 0 1 1 1 1 1 0 1 1 1 1\\
0 0 0 0 1 0 0 0 0 0 0 0 0 0 0 1 0 0 0 1 0 1 1 0 1 1 1 0 1 1 1 1\\
\end{array}
\right]
\]
}

\comment{

%Part1 of Main Table
\begin{table}
 \caption{ODP for Near Extremal FSD Even $[16,8,4]$ codes (Part 1)}
 \label{tab:fsd16_8_4_1}
 \begin{center}
{\small
\begin{tabular}{|l|c|c|c|c|c|}
\noalign{\hrule height1pt}
Weight Distribution & ODP & SD & ID & FSD & Total  \\
       \hline

$\begin{array}{l}A_{0,16}= 1, A_{4,12}= 4,\\ A_{6,10}= 96, A_{8}= 54 \end{array}$ & ${ODP}=[4, 6, 6, 8, 8, 8, 8, 16]$ & 0 & 1 & 0 & 1\\
\hline
$\begin{array}{l}A_{0,16}= 1, A_{4,12}= 8, \\ A_{6,10}= 80, A_{8}= 78\end{array}$ & $ \begin{array}{l}{ODP}^{dic}=[4,4,6,6,6,6,10,12] \\ {ODP}^{inv}=[4,4,6,6,6,8,8,16] \end{array}$  & 0 & 1 & 0 &  \\ \cline{2-5}
& ${ODP}=[4,4,6,6,8,8,8,16]$ & 0 & 3 & 0 &  \\ \cline{2-5}
 & $ {ODP}=[4,6,6,8,8,8,8,16] $ & 0 & 1 & 0 & 5 \\
\hline
$\begin{array}{l}A_{0,16}= 1, A_{4,12}= 10, \\ A_{6,10}= 72, A_{8}= 90\end{array}$ & $ \begin{array}{l}{ODP}^{dic}=[4,4,4,6,6,8,10,12] \\ {ODP}^{inv}=[4,4,4,6,6,8,8,16] \end{array}$ & 0 & 0 & 1 & \\ \cline{2-5}
 & $ \begin{array}{l}{ODP}^{dic}=[4,4,6,6,6,8,10,12] \\ {ODP}^{inv}=[4,4,6,6,6,8,8,16] \end{array}$  & 0 & 0 & 7 &  \\ \cline{2-5}
 & $ {ODP}=[4,4,6,6,8,8,8,16] $ & 0 & 6 & 2 &  \\ \cline{2-5}
 & $ \begin{array}{l}{ODP}^{dic}=[4,6,6,6,8,8,8,10] \\ {ODP}^{inv}=[4,4,6,6,8,8,8,16] \end{array}$  & 0 & 1 & 0 & 17 \\
\hline
$\begin{array}{l}A_{0,16}= 1, A_{4,12}= 12, \\ A_{6,10}= 64, A_{8}= 102\end{array}$ & $ \begin{array}{l}{ODP}^{dic}=[4,4,4,6,6,8,10,12] \\ {ODP}^{inv}=[4,4,4,6,6,8,8,16] \end{array}$ & 0 & 0 & 1 & \\ \cline{2-5}
 & $ {ODP}=[4,4,4,6,8,8,8,16] $ & 0 & 2 & 2 &  \\ \cline{2-5}
 & $ \begin{array}{l}{ODP}^{dic}=[4,4,6,6,6,8,10,12] \\ {ODP}^{inv}=[4,4,4,6,8,8,8,16] \end{array}$  & 0 & 2 & 1 &  \\ \cline{2-5}
 & $ \begin{array}{l}{ODP}^{dic}=[4,4,6,6,6,8,10,12] \\ {ODP}^{inv}=[4,4,6,6,6,8,8,16] \end{array}$  & 0 & 1 & 1 &  \\ \cline{2-5}
 & $ {ODP}=[4,4,6,6,8,8,8,16] $ & 0 & 14 & 25 &  \\ \cline{2-5}
 & $ {ODP}=[4,4,6,8,8,8,8,16] $ & 1 & 5 & 2 &  \\ \cline{2-5}
 & $ {ODP}=[4,6,6,8,8,8,8,16] $ & 0 & 1 & 0 & 58 \\
\hline

 \noalign{\hrule height1pt}
\end{tabular}
}
\end{center}
\end{table}

%Part2 of Main Table

\begin{table}
 \caption{ODP for Near Extremal FSD Even $[16,8,4]$ codes (Part 2)}
 \label{tab:fsd16_8_4_2}
 \begin{center}
{\small
\begin{tabular}{|l|c|c|c|c|c|}
\noalign{\hrule height1pt}
Weight Distribution & ODP & SD & ID & FSD & Total  \\
       \hline
$\begin{array}{l}A_{0,16}= 1, A_{4,12}= 14, \\ A_{6,10}= 56, A_{8}= 114\end{array}$ & ${ODP}=[4, 4, 4, 6, 8, 8, 8, 16]$ & 0 & 0 & 4 & \\ \cline{2-5}
 & $ {ODP}=[4,4,4,8,8,8,8,16] $ & 0 & 2 & 2 &  \\ \cline{2-5}
 & $ \begin{array}{l}{ODP}^{dic}=[4,4,6,6,6,8,10,12] \\ {ODP}^{inv}=[4,4,4,6,8,8,8,16] \end{array}$  & 0 & 0 & 2 &  \\ \cline{2-5}
 & $ \begin{array}{l}{ODP}^{dic}=[4,4,6,6,8,8,8,10] \\ {ODP}^{inv}=[4,4,4,6,8,8,8,16] \end{array}$  & 0 & 0 & 8 &  \\ \cline{2-5}
 & $ \begin{array}{l}{ODP}^{dic}=[4,4,6,6,8,8,8,10] \\ {ODP}^{inv}=[4,4,4,8,8,8,8,16] \end{array}$  & 0 & 1 & 0 &  \\ \cline{2-5}
 & $ \begin{array}{l}{ODP}^{dic}=[4,4,6,6,8,8,8,12] \\ {ODP}^{inv}=[4,4,4,6,8,8,8,16] \end{array}$  & 0 & 1 & 1 &  \\ \cline{2-5}
 & $ \begin{array}{l}{ODP}^{dic}=[4,4,6,6,8,8,8,16] \\ {ODP}^{inv}=[4,4,4,8,8,8,8,16] \end{array}$  & 0 & 6 & 3 &  \\ \cline{2-5}
 & $ {ODP}=[4,4,6,6,8,8,8,16] $ & 0 & 2 & 4 & 36 \\

 \noalign{\hrule height1pt}
\end{tabular}
}
\end{center}
\end{table}

%Part3 of Main Table

\begin{table}
 \caption{ODP for Near Extremal FSD Even $[16,8,4]$ codes (Part 3)}
 \label{tab:fsd16_8_4_3}
 \begin{center}
{\small
\begin{tabular}{|l|c|c|c|c|c|}
\noalign{\hrule height1pt}
Weight Distribution & ODP & SD & ID & FSD & Total  \\
       \hline

$\begin{array}{l}A_{0,16}= 1, A_{4,12}= 16, \\ A_{6,10}= 48, A_{8}= 126\end{array}$ & $ \begin{array}{l}{ODP}^{dic}=[4,4,4,6,6,8,10,12] \\ {ODP}^{inv}=[4,4,4,6,6,8,8,16] \end{array}$ & 0 & 1 & 0 & \\ \cline{2-5}
 & $ {ODP}=[4,4,4,6,8,8,8,16] $ & 0 & 2 & 2 &  \\ \cline{2-5}
 & $ {ODP}=[4,4,4,8,8,8,8,16] $ & 0 & 3 & 3 &  \\ \cline{2-5}
 & $ \begin{array}{l}{ODP}^{dic}=[4,4,6,6,8,8,8,12] \\ {ODP}^{inv}=[4,4,4,6,8,8,8,16] \end{array}$  & 0 & 0 & 1 &  \\ \cline{2-5}
 & $ \begin{array}{l}{ODP}^{dic}=[4,4,6,6,8,8,8,12] \\ {ODP}^{inv}=[4,4,4,8,8,8,8,16] \end{array}$  & 0 & 0 & 1 &  \\ \cline{2-5}
 & $ \begin{array}{l}{ODP}^{dic}=[4,4,6,6,8,8,8,16] \\ {ODP}^{inv}=[4,4,4,8,8,8,8,16] \end{array}$  & 0 & 2 & 1 &  \\ \cline{2-5}
 & $ {ODP}=[4,4,6,6,8,8,8,16] $ & 0 & 1 & 2 &  \\ \cline{2-5}
 & $ {ODP}=[4,4,6,8,8,8,8,16] $ & 0 & 2 & 0 & 21 \\
\hline
$\begin{array}{l}A_{0,16}= 1, A_{4,12}= 20, \\ A_{6,10}= 32, A_{8}= 150\end{array}$ &  $ {ODP}=[4,4,4,8,8,8,8,16] $ & 0 & 3 & 0 & \\ \cline{2-5}
 & $ {ODP}=[4,4,6,8,8,8,8,16] $ & 0 & 1 & 0 & 4 \\
\hline
$\begin{array}{l}A_{0,16}= 1, A_{4,12}= 28, \\ A_{8}= 198\end{array}$ & $ {ODP}=[4,4,4,8,8,8,8,16] $ & 2 & 0 & 0 & 2\\
\hline
{\bf Total} & & 3 & 65 & 76 & 144\\

 \noalign{\hrule height1pt}
\end{tabular}
}
\end{center}
\end{table}

}%end comment

%\newpage


\bigskip

\begin{thebibliography}{}
%\bibitem{Kro} http://en.wikipedia.org/wiki/Kronecker\_symbol.

%\bibitem{sdtables}
%http://www.unilim.fr/pages\_perso/philippe.gaborit/SD/.

%\bibitem{BetHara} K. Betsumiya and M. Harada, Binary Optimal Odd Formally Self-Dual Codes, Designs, Codes and Cryptography, 23, No. 1, (2001), pp. 11--21.

%\bibitem{BetHara2} K. Betsumiya and M. Harada, Classification of Formally Self-Dual Even Codes of Lengths up to 16 , Designs, Codes and Cryptography, 23, No. 3, (2001), pp. 325--332.

\bibitem{Har2011} Koichi Betsumiya, Masaaki Harada, Akihiro Munemasa.  A complete classification of doubly-even self-dual
codes of length 40. Online available at http://arxiv.org/pdf/1104.3727v2.pdf .  2011.

\bibitem{DijBagTol} M. van Dijk, S. Baggen, and L. Tolhuizen, Coding
for informed decoders, in Proc. IEEE Int. Symp. Inf. Theory,
Washington, DC. Ju. (2001), pp. 202.

\bibitem{BouJaffe} I. Bouyuklieva and D.B. Jaffe, Optimal binary linear codes of dimension at most seven,
Discrete Math, Vol. 226, (2001), pp. 51--70.

\bibitem{Mag} J. Cannon and C. Playoust,
 \textit{An Introduction to Magma}, University of Sydney, Sydney,
 Australia, (1994), version V2.12-19.

\bibitem{CheVin_2010} Y. Chen and A.J. Han Vinck, A lower bound on the optimum distance profiles of the second-order Reed-Muller codes, IEEE Trans. Inform. Theory, 56, No.9, (2010), pp. 4309--4320.


\bibitem{ConPle} J.H. Conway and V. Pless, On the enumeration of
self-dual codes, J. Combin. Theory. Ser. A, Vol. 28, No. 1, (1980),
pp. 26--53.

\bibitem{ConPleSlo} J.H. Conway, V. Pless, and N.J.A. Sloane, The
binary self-dual codes of length up to $32$:A revised enumeration,
J. Combin. Theory. Ser. A, Vol. 60, No. 2, (1992), pp. 183--195.

\bibitem{DodEnKa} S. M. Dodunekov, S. B. Encheva and S. N. Kapralov.  On the $[28, 7, 12]$ binary self-complementary codes and their residuals, Designs, Codes and Cryptography,  Vol. 4, No. 1, (1994), pp. 57--67.

%\bibitem{FieGabHuffPless} J.E. Fields, P. Gaborit, W.C. Huffman, and V. Pless.  On the classification of extremal even formally self-dual codes of lengths 20 and 22, Discrete Applied Math, Vol. 111, Iss. 1--2, (2001), pp. 75--86.

\bibitem{codetables} M. Grassl, {\it Bounds on the minimum distance of linear codes},
online available at http://www.codetables.de .

\bibitem{GulOst} T.Aaron Gulliver and P.R.J. \"{O}stergard, Binary optimal linear rate 1/2 codes, Discrete Math, Vol. 283, Iss. 1--3, (2004), pp. 255--261.

\bibitem{HolTos} H. Holma and A. Toskala, WCDMA for UMTS-HSPA
Evolution and LTE, 4th ed. London, U.K.: Wiley, 2007.

\bibitem{HufPle} W.C. Huffman and V. Pless,
Fundamentals of Error-Correcting Codes, Cambridge University Press,
2003.

\bibitem{Jaf} D.B. Jaffe, Optimal binary linear codes of length $\le
30$, Discrete Math, Vol. 223, (2000), pp. 135--155.

\bibitem{JafCodeform} D.B. Jaffe, {\it Information about binary linear codes},
online available at http://www.math.unl.edu/~djaffe2/codes/webcodes/codeform.html .

\bibitem{Koh} A. Kohnert, New $[48,16,16]$ optimal linear binary block code, Online available at http://arXiv:0912.4107v1 .  2009.

\bibitem{LuoVinChe_2010} Y. Luo, A.J. Han Vinck, Y. Chen, On the optimum distance profiles about linear block codes, IEEE Trans. Inform. Theory, 56, No.3, (2010), pp. 1007--1014.



\bibitem{MacSlo} F.J.MacWilliams and N.J.A.Sloane, The Theory of Error-Correcting Codes (North-Holland, 1977).

\bibitem{MakSim} J. Maks and J. Simonis, Optimal subcodes of second
order Reed-Muller codes and maximal linear spaces of bivectors of
maximal rank, Des. Codes and Cryptogr., 21, (2000), pp. 165--180.


\bibitem{MilStu} E. Miller and B. Sturmfels, {\em Combinatorial Commutative Algebra}, GTM, Springer, New York, 2005.

\bibitem{Ple_72} V. Pless, A classification of self-orthogonal codes over $GF(2)$,
Discrete Math, 3, (1972), pp. 215--228.

\bibitem{PleHuf} V. S. Pless, W.C. Huffman, Eds. {\em Handbook of Coding Theory}, Amsterdam. The Netherlands: Elsevier, 1998.

\bibitem{PleSlo} V. Pless and N.J.A. Sloane, On the
classification and enumeration of self-dual codes, J. of Combin.
Theory, 18, No. 3, (1975), pp. 313--335.


\bibitem{PolCheMcE} F. Pollara, K.-M. Cheung, and R.J. McEliece,
Further results on finite-state codes, TDA Progress Report 42-92,
October-Decemeber, (1987), pp. 56--62.

\bibitem{RS} E. M. Rains and N. J. A. Sloane, `` Self-dual codes,''
in {\em Handbook of Coding Theory}, ed. V. S. Pless and W. C.
Huffman.  Amsterdam: Elsevier, pp. 177--294, 1998.

\bibitem{Simonis} Juriaan Simonis.  The [18, 9, 6] code is unique, Discrete Math, Vol. 106-107, 1, (1992), pp. 439--448.

\bibitem{TanWoo} R. Tanner and J. Woodard, WCDMA-Requirements and
Practical Design, London, U.K.: Wiley, 2004.

\bibitem{Til} H. van Tilborg, On the uniqueness resp. nonexistence of certain codes meeting the Griesmer bound, Inf. Control. 44, (1980), pp. 16--35.

\bibitem{YanZhuLuo} J. Yan, Z. Zhuang, and Y. Luo, On the optimum distance profiles of some quasi cyclic codes, Proc. of 2011 13th International Conference on Communication Technology, (2011), pp. 979--983.


%\bibitem{DouGulHar} S. T. Dougherty, T. A. Gulliver, and M. Harada,
%{\it Extremal binary self-dual codes}, IEEE Trans. Inform. Theory,
%{\bf {43}} (1997), 2036--2047.

%\bibitem{DouKimSol} S. T. Dougherty, J.-L. Kim, and P. Sol\'{e},
%{\it Double circulant codes from two class association schemes},
%Advences in Math. of Communications, {\bf{1}} (2007), 45--64.

%\bibitem{Gab} P. Gaborit,
%{\it Quadratic double circulant codes over fields}, J. Combin.
%Theory, Ser. A, \textbf{97} (2002), 85--107.

%\bibitem{GabOtm} P. Gaborit and A. Otmani,
%{\it Experimental constructions of self-dual codes},
%Finite Fields Appl., \textbf{9} (2003), 372--394.


%\bibitem{GooYor} V. Goodwin and V. Yorgov,
%{\it New extremal self-dual doubly-even binary codes of length
%$88$}, Finite Fields Appl., \textbf{11} (2005), 1--5.


%\bibitem{GulHar2} T. A. Gulliver and M. Harada,
%{\it Classifiation of extremal double circulant formally %self-dual
%even codes}, Des. Codes and Cryptogr. 11 (1997), 25-35.

%\bibitem{GulHar} T. A. Gulliver and M. Harada,
%{\it Classification of extremal double circulant self-dual codes
%of lengths 64 to 72}, Des. Codes and Cryptogr., \textbf{13}
%(1998), 257--269.

%\bibitem{GulHar3} T. A. Gulliver and M. Harada,
%{\it New nonbinary self-dual codes}, IEEE Trans. Inform. Theory,
%{\bf {54}} (2008), 415--417.

%\bibitem{GulHarMiy}T. A. Gulliver, M. Harada, and H. Miyabayashi,
%{\it Double circulant self-dual codes over $\mathbb F_5$ and $\mathbb F_7$},
% Advances in Math. of Communications, \textbf{1} (2007), 223--238.

%\bibitem{GulSen} T. A. Gulliver, Nikolai Senkevitch, On a Class of Self-Dual Codes Derived from %Quadratic Residues,
%IEEE Trans. Inform. Theory IT-45 (1999) 701--702.

%\bibitem{HanLee} S. Han, J.B. Lee, Nonexistence of Some Extremal Self-Dual Codes,
%J. Korean Math. Soc. 43 (2006), No. 6, pp. 1357-1369.

%\bibitem{HanKim} S. Han and J.-L. Kim,
%{\it On self-dual codes over $\mathbb F_5$}, Des. Codes and
%Cryptogr., \textbf{48} (2008), 43--58.

%\bibitem{HarNis} M. Harada, T. Nishimura, An extremal singly even self-dual code of
%length 88, Advances in Mathematics of Communications, Volume 1, No. 2, 2007, 261--267.

%\bibitem{HarMun} M. Harada and A. Munemasa,
%{\it There exists no self-dual $[24,12,10]$ code over $\F_5$}, Des. Codes and Cryptogr., to appear.

%\bibitem{HufPle} W. C. Huffman and V. S. Pless,
%{\it Fundamentals of Error-Correcting Codes}, Cambridge: Cambridge
%University Press (2003).

%\bibitem{Kar} M. Karlin, {\it New binary coding results by circulants},
%IEEE Trans. Inform. Theory, \textbf{15} (1969), 81--92.

%\bibitem{Huf} W. C. Huffman, On the
%classification and enumeration of self-dual codes, Finite Fields
%Appl. 11 (2005) 451-490.


%\bibitem{LMP} J. S. Leon, J. M. Masley, V. Pless, Duadic codes, IEEE Trans. Inform. Theory IT-30
%(1984), 709--714.

%\bibitem{Ple} V. Pless, {\it Q-codes}, J. Combin. Theory, Ser. A,
%\textbf{43} (1986), 258--276.

%\bibitem{Ple72} V. Pless, {\it Symmetry codes over GF($3$) and new five-designs}, J.
%Comb. Theory, Ser. A, \textbf{12} (1972), 119--142

\end{thebibliography}
\end{document}